\documentclass[journal]{IEEEtran}
\ifCLASSINFOpdf
\else
\fi

\usepackage{amssymb}
\usepackage{amsthm}
\usepackage[cmex10]{amsmath}
\usepackage{amsfonts}
\usepackage{dsfont}
\usepackage{graphicx}
\usepackage[caption=false]{subfig}
\usepackage{cite}
\usepackage{url}
\usepackage{algorithm}
\usepackage{algpascal}
\usepackage{algpseudocode}
\usepackage{array}

\def\be{\begin{equation}}
\def\ee{\end{equation}}
\def\ben{\begin{equation*}}
\def\een{\end{equation*}}
\def\bea{\begin{eqnarray}}
\def\eea{\end{eqnarray}}
\def\beaa{\begin{eqnarray*}}
\def\eeaa{\end{eqnarray*}}

\def\tb{\textbf}
\def\mb{\mathbf}

\def\re{\mathfrak{Re}}
\def\im{\mathfrak{Im}}

\theoremstyle{remark} \newtheorem{lemma}{Lemma}
\theoremstyle{remark} 
\theoremstyle{remark} 

\algblock[Name]{Start}{End}
\algblockdefx[NAME]{START}{END}%
[2][Unknown]{Start #1(#2)}%
{Ending}
\algblockdefx[NAME]{}{OTHEREND}%
[1]{Until (#1)}

\usepackage{color}
\begin{document}

\title{Spatio-Spectral Radar Beampattern Design for Co-existence with Wireless Communication Systems }

\author{Bosung~Kang,~\IEEEmembership{Member,~IEEE,}
	Omar~Aldayel,~%\IEEEmembership{Student~Member,~IEEE,}
	Vishal~Monga,~\IEEEmembership{Senior~Member,~IEEE,}
	and Muralidhar~Rangaswamy,~\IEEEmembership{Fellow,~IEEE}
	 \thanks{Research was supported by AFOSR grant number FA9550-15-1-0438.
	
	 Dr. Rangaswamy was supported by the Air Force Office of Scientific Research under project LRIR 17 RYCOR 481.}% <-this % stops a space
}

\maketitle

%\markboth{Accepted to IEEE Transactions on Aerospace and Electronic Systems}{Kang \MakeLowercase{\textit{et al.}}: Spatio-Spectral Radar Beampattern Design}

\begin{abstract}
 We address the problem of designing a transmit beampattern for multiple-input multiple-output (MIMO) radar considering co-existence with wireless communication systems. The designed beampattern is able to manage the transmit energy in \emph{spatial} directions as well as in \emph{spectral} frequency bands of interest by minimizing the deviation of the designed beampattern versus a desired one under a spectral constraint as well as the constant modulus constraint. While unconstrained beampattern design is straightforward, a key open challenge is jointly enforcing the spectral  constraint in addition to the constant modulus constraint on the radar waveform. A new approach is proposed in our work, which involves solving a sequence of constrained quadratic programs such that constant modulus is achieved at convergence. Further, we show that each problem in the sequence has a closed form solution leading to analytical tractability. We evaluate the proposed beampattern  with interference control (BIC) algorithm against the state-of-the-art MIMO beampattern design techniques and show that BIC achieves closeness to an idealized beampattern along with desired spectral shaping.
\end{abstract}



\begin{IEEEkeywords}
MIMO radar, beampattern design, spectral constraint, constant modulus, successive algorithm, waveform design, closed form solution, spectral co-existence
\end{IEEEkeywords}

% \fi
\IEEEpeerreviewmaketitle

% \input{section1}
% \input{section2}
% \input{section3}
% \input{section4}
%\section*{Acknowledgment}

\section{Introduction}
\label{Sec:Introduction}
In wideband radar applications such as the high-resolution and ultra wideband (UWB) noise radars, the radar system requires a large bandwidth.  For example, in microwave systems and UWB noise radar, the waveform bandwidth is about 1 GHz, while in ultra high frequency (UHF) systems the waveform bandwidth can exceed 200 MHz \cite{Nunn12,Rowe14,Surender10}. In these applications, radar emissions will overlap with the spectrum allocated for communications and other wireless systems. Co-existence of radar and telecommunication systems has been an emerging requirement recently \cite{Aubry14AES,Aubry14Radarcon,Shepherd13,Guo14,Xin14,Huleihel13,Tang10,Amuru13,Aubry15,Guerci15,Lackpour16}. A priori knowledge about expected target locations and the radio frequency environment enables MIMO radar systems to enhance the probability of detection while ensuring compatibility with civilian wireless systems. Specifically, the MIMO radar should focus the radiation beam in the expected target directions while maintaining a low spectral interference level at specific bands used by other licensed wireless systems. These two objectives can be achieved by constrained optimization of the radar transmit waveform \cite{Skolnik90,Gini12}. 

When it comes to radar beampattern optimization/design problem, two main research directions have been actively pursued to ensure co-existence of radar and communication systems in the past years. First, optimization of MIMO radar waveform to match the desired beampattern with an arbitrary spectrum shape has been a topic of much recent interest \cite{Guo15,Wang12,Zang14,Sen13,Ahmed14,Zhang15,Pan16,Chen14,Stoica08,Hua13,Xu15,He11,Aldayel17TSP,Roberts08,Deng13,Fuhrmann08,Stoica08WS,Ahmed11,Aubry16MR}. In these methods, the goal of the optimization problem is to minimize deviation of the optimized beampattern to the desired one which is designed to reduce the transmit energy at \emph{spatial} angles where the communication systems are located. Some of these works focus on receive beampattern design \cite{Chen14,Roberts08,Deng13} while most others focus on the transmit beampattern design. On the other hand, mitigation of the energy of the transmit waveform in the \emph{spectral} frequency bands occupied by wireless communication systems has also been studied \cite{Rowe14}. This approach matches the spectral shape of the optimized waveform to the desired one which is designed to limit the interference level on communication systems or to directly minimize the interference level at communication receivers. However, since the beampattern is not considered, it is not able to control the radiation beam in spatial directions.

\subsection{Motivation and Challenges}
In practice, the transmit beampattern design is more challenging for two reasons. The first reason is the requirement of the constant modulus constraint on the radar transmit waveform, i.e. a constant envelope transmit signal \cite{Patton08}. The importance of the constant modulus waveform has been well documented and analyzed in terms of performance loss \cite{Patton08,PattonThesis,Friedlander07}. A non-linear power amplifier which is equipped in most radar systems cannot be efficiently utilized without the constant modulus constraint since the output of the amplifier will be a clipped version of the optimized waveform. The second reason is the requirement of spectral compatibility of radar and telecommunication systems, which demands a spectral constraint on the radar waveform spectral shape. Designing the MIMO radar beampattern in the simultaneous presence of constant modulus and spectral constraints remains a stiff open challenge.

%It is well known that the problem of waveform design subject to the constant modulus constraint constitutes a hard non-convex problem. To ensure tractability, some existing approaches pursue relaxations to energy constraint (using $L2$ norm) \cite{Aubry14AES,Aubry16} or approximations to the constant modulus constraint \cite{Stoica08,He11,Rowe14}. This indirect approximation makes the problem more tractable, however, it degrades the design accuracy. Some recent efforts directly enforce the constant modulus constraint and hence lead to better performance. However, they invariably involve semi-definite relaxation (SDR) with randomization \cite{Cui14TSP,Luo10}. In this approach, a semi-definite programming (SDP) is first solved to find a waveform distribution. Then a large number of trials are generated based on this distribution, which is followed by exhaustive search to find the best waveform. Although SDR with randomization gives a good approximate solution with the constant modulus constraint, there is no guarantee of a reasonable solution.

It is well known that the MIMO transmit beampattern/waveform design subject to the constant modulus constraint constitutes a hard non-convex problem. To ensure tractability, some existing approaches pursue relaxations to energy constraint (using $L2$ norm) \cite{Aubry14AES,Aubry16} or approximations to the constant modulus constraint \cite{Stoica08,He11,Rowe14}. This indirect approximation makes the problem more tractable, however, it degrades the design accuracy. Some recent efforts directly enforce the constant modulus constraint and hence lead to better performance. However, they invariably involve semi-definite relaxation (SDR) with randomization \cite{Cui14TSP,Luo10}. In this approach, a semi-definite programming (SDP) is first solved to find a waveform distribution. Then a large number of random waveforms are generated based on this distribution, which is followed by an exhaustive search to find the closest waveform. Despite the success of SDR for constant modulus constrained problems, two issues remain: 1.) extensions to spectral constraints, which are quadratic inequalities are {\em not} straightforward, and 2.) the computational burden is high.

Beampattern design under the constant modulus constraint but without the spectral  constraint has been studied in \cite{Guo15,Zang14,Wang12,He11,Stoica09SPL,Aldayel17TSP}. In the beampattern design problems, an approximation to constant modulus was pursued using the peak-to-average power ratio (PAPR) waveform constraint \cite{Stoica08,He11}. While the constant modulus constraint is not explicitly represented in the optimization process, the resulting solution is converted to the nearest constant modulus solution.

Similarly, as mentioned before there is active interest in radar-comm co-existence where the transmit waveform is optimized but without the constant modulus constraint \cite{Aubry14AES, Shepherd13,Guo14,Xin14, Aubry15,Guerci15,Lackpour16}. 

Of particular interest is the  recent work of Guerci \emph{et al.} which presents a new paradigm for the joint design and operation (JDO) of shared spectrum access for radar and communications (SSPARC) \cite{Guerci15}. They optimize transmit waveforms at both radar and communication nodes in a way that maximizes the signal power through the forward channels (resp. of radar and communication systems) while simultaneously minimizing the response in the co-channels between radar and communications. This optimization can be extended to achieve a low probability of intercept capability in specific angular keep-out zones where co-channel RF nodes are located. A crucial difference of our proposal from Guerci's work (and others that design beampatterns for spectral co-existence) is that they consider a shared-spectrum scenario and hence the design is spatially based. Whereas, we consider spatio-spectral design where the frequency spectrum of the transmit radar waveform is explicitly shaped. Further, a constant modulus constraint is enforced for practical deploy-ability.
%On the other hand, to ensure both co-existence of radar and communication systems and constant modulus, the objective of some recent work is to minimize radiation power in a few selected directions while the design criterion does not allow full control of power allocation. An interesting recent advance that forces the constant modulus constraint for beampattern design has been proposed in \cite{Guo15} which sets up waveform design as a phase optimization problem and solves it using a typical iterative numerical method but with no known analytical guarantees of the resulting solution.

% In this work, we develop a new algorithm to design radar waveforms for transmit beamforming for both narrowband and wideband MIMO systems termed Sequence of Closed Forms (BIC) algorithm. Furthermore, the algorithm

% On the other hand, the Similarity Constraint (SC), uses a reference signal as a benchmark to produce an optimized waveform that shares some of the desirable autocorrelation properties of the reference waveform. As noted in \cite{friedlander2007waveform,chen2009mimo}, the resulting waveforms from algorithms that do not enforce SC suffer from undesirable artifacts in pulse compression and ambiguity function properties.
%a gradient-based methods  \cite{wang2012design}. However, since the problem is hard non-convex problem, local optimum solutions might be
%Moreover, the performance some of the state of the art methods \cite{guo2015waveform,wang2012design} is highly dependent on the initial point and the step length.

\subsection{Our contributions}
\label{Sec:Contribution}
Our principal aim is to develop an algorithmic approach for {\em spatio-spectral} MIMO beampattern design. Closeness to an idealized beampattern that limits radar energy in the direction of wireless communication receivers captures the spatial component while the spectral component of our approach involves explicitly forcing a spectral fidelity constraint.

Specifically, this paper makes the following contributions:
\begin{itemize}
	\item \textbf{A new algorithmic solution for spatio-spectral beampattern design under both the spectral  constraint and the constant modulus constraint.} To overcome the challenges mentioned above, we develop a new algorithm for MIMO beampattern design that involves solving the hard non-convex problem of beampattern design using a sequence of convex equality and inequality constrained quadratic programs (QP), each of which has a closed form solution, such that constant modulus is achieved at convergence. Because each QP in the sequence has a closed form solution, the proposed beampattern with interference control (BIC) algorithm has significantly lower complexity than 
	most competing methods.
%	SDR with relaxation \cite{Cui14TSP,Luo10} which is a representative algorithm for the constant modulus constraint.
	
	\item \textbf{Feasibility of the sequence of QPs.} Assuming that the original non-convex problem of beampattern design is feasible, i.e. the intersection set of constant modulus and spectral constraints is non-empty; we formally prove that each QP we formulate in the aforementioned BIC sequence is also guaranteed to be feasible.
	\item \textbf{Convergence of the BIC algorithm.} We establish that the sequence of cost functions representing a deviation from the desired beampattern, that occurs in the proposed BIC algorithm, is non-increasing, (i.e. an improvement is always obtained by solving each problem in the sequence) and converges.
	
	\item \textbf{Experimental insights and validation.} Experimental validation is performed across two scenarios: 1) null forming where the BIC algorithm shows significant power suppression in the desired directions even in the presence of the spectral constraint, and 2) full beampattern design where the proposed BIC is shown to achieve a beampattern much closer to the ground truth against state of the art alternatives that have no spectral interference constraint.
\end{itemize}

The rest of the paper is organized as follows. Section \ref{Sec:SystemModel} provides brief background on the structure of the radar antenna array and the corresponding design criterion. Section \ref{Sec:BIC} develops the proposed BIC algorithm for the two cases of wideband beampattern design and nullforming beampattern design and reports derivations of its analytical properties. Section \ref{Sec:Results} evaluates the proposed BIC method against state-of-the-art alternatives. Concluding remarks with directions for future work are presented in Section \ref{Sec:Conclusion}.

\subsection{Notation}

We denote vectors and matrices by boldface letters, e.g. $\mb a$ (lowercase) and $\mb A$ (uppercase), respectively. The $l$-th element of $\mb a$ is denoted by $\mb a_l$ and the element located in the \emph{m}-th row and \emph{l}-th column of the matrix $\mb A$ is denoted by $\mb A(m,l)$. We denote by $\| \mb a\|_2$ the $l_2$ norm of the vector $\mb a$. The Hermitian, conjugate and transpose operators are denoted by $(.)^H$, $(.)^*$ and $(.)^T$, respectively. For a complex number $a$, we denote $\re (a)$ and $\im (a)$ to the real and imaginary part $a$, respectively; also we denote $|a|$ and $\arg a$ to the amplitude and phase of $a$, respectively. We use $j = \sqrt{-1}$ as the imaginary unit number. Finally, we use $\otimes$ to denote the Kronecker product.

\section{System Model}
\label{Sec:SystemModel}
\begin{figure}
\centering
\includegraphics[scale=0.22]{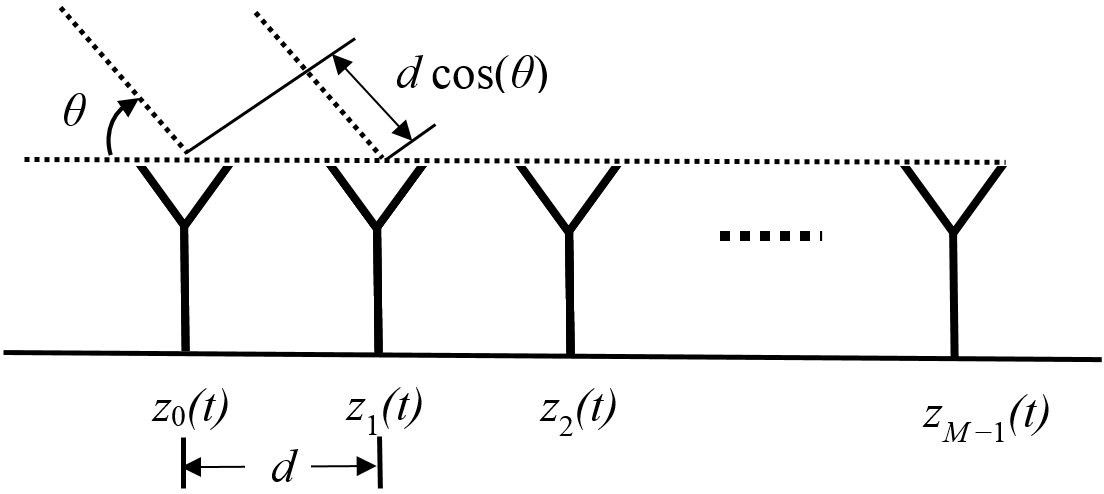}
\caption{Configuration of ULA antenna}
\label{Fig:ULA}
\end{figure}

% In the literature, the problem of designing a MIMO beam pattern can be formulated in two different w
% In \cite{he2011wideband},

Consider a wideband MIMO radar with a uniform linear array (ULA) of $M$ antennas and equal spacing distance of $d$ as shown in Fig. \ref{Fig:ULA}. The signal transmitted from the $m$-th element is denoted by $z_m(t)$. Let $z_m(t)=x_m(t) e^{j 2 \pi f_c t}$ where $x_m(t)$ is the baseband signal and $f_c$ is the carrier frequency. We assume that the spectral support of $x_m(t)$ is within the interval $[-B/2, B/2]$ where $B$ is the bandwidth in Hz. The sampled baseband signal transmitted by the $m$-th element is denoted by $x_m(n)\triangleq x_m(t=nT_s)$, $n=0,...,N-1$ with $N$ being the number of time samples and $T_s=1/B$ is the sampling rate. The discrete Fourier transform (DFT) of $x_m(n)$ is denoted by $y_m(p)$ and it is given by
\be
\label{eq:ym}
y_m(p)=\sum_{n=0}^{N-1} x_m(n) e^{-j2\pi \frac{np}{N}}, \quad p= -\frac{N}{2},\ldots,0,\ldots, \frac{N}{2}-1 
\ee
where $N$ is assumed to be even\footnote{Note that we assume that $N$ is even in this paper without loss of generality.} in Eq. (\ref{eq:ym}). If $N$ is odd, then $p= -(N-1)/2,\ldots,0,\ldots, (N-1)/2$.

\subsection{Far-Field Beampattern}
\label{Sec:FFBP}

According to \cite{He11}, the discrete frequency beampattern at the angle $\theta$ in the frequency band $p$ in the far-field is given by
\be
P(\theta, p)=|\mb a^H(\theta,p)\mb y_p|^2
\ee
where
\be
\mb a(\theta, p)=[1\quad e^{j2\pi(\frac{p}{NT_s}+f_c)\frac{d\cos\theta}{c}} \, \ldots \, e^{j2\pi(\frac{p}{NT_s}+f_c)\frac{(M-1)d\cos\theta}{c}}]^T
\ee
and
\be
\mb y_p=[y_0(p) \quad y_1(p) \quad ... \quad y_{M-1}(p)]^T
\ee
where $c$ is the speed of wave propagation. Note that $\mb a(\theta, p)$ is continuous in phase. It can be expressed as a discrete angle vector by dividing the interval $[0^{\circ}, 180^{\circ}]$ into $K$ angle bins. Using the same simplified notation found in \cite{He11}, it can be written as
\be
\mb a_{kp}=\mb a(\theta_k,p), \quad k=1, 2, \ldots, K
\ee
In this case, the beampattern can be given by the following discrete angle-frequency grid
\be
P_{kp}=|\mb a^H_{kp} \mb y_p|^2=|\mb a^H_{kp} \mb W_p \mb x|^2
\ee
where $\mb x \in \mathbb{C}^{MN}$ is the concatenated vector i.e. $\mb x=[\mb x^T_0 \quad \mb x^T_1 \quad ... \quad \mb x^T_{M-1}]^T$ where  $\mb x_{m}=[x_m(0) \quad x_m(1) \quad \ldots \quad x_m(N-1)]^T \in \mathbb{C}^{N}$ and $\mb W_p \in \mathbb{C}^{M \times MN}$ is given by
\be
\mb W_p= \mb I_M  \otimes \mb e^H_p
\ee
where $\mb e^H_p=[1 \quad e^{-j2\pi \frac{p}{N}} \quad \ldots \quad  e^{-j2\pi \frac{(N-1)p}{N}}] \in \mathbb{C}^{N}$ and $\mb I_M$ is an $M\times M$ identity matrix.

\subsection{Formulation of the Spectral Constraint}
\label{Sec:SI}
The problem of spectral co-existence has been of great interest recently \cite{Aubry14AES,Aubry14Radarcon,Shepherd13,Guo14,Xin14,Huleihel13,Tang10,Amuru13,Aubry15} and involves minimization of interference caused by radar transmission at victim communication receivers operating in the same frequency band. In this case, the beampattern of the transmit waveform is required to have nulls in these bands to prevent interference. For $J$ communication receivers, we suppose that the $j$-th communication receiver operating on a frequency band $B_j=[p^j_l, p^j_u]$, where $p^j_l$ and $p^j_u$ are the lower and upper normalized frequency, respectively. We denote the desired (discrete) spectrum shape by $ \hat{\mb y}=[\hat{y}_{-\frac{N}{2}}, \hat{y}_{-\frac{N}{2}+1}, ...,\hat{y}_{\frac{N}{2}-1}] \in \mathbb{C}^{N\times 1}$ defined as

\begin{equation*}
\hat{y}_p=
\begin{cases}
0 & \text{for } p \in B_j=[p^j_l, p^j_u], \qquad j=1,2,..., J\\
\gamma  & \text{otherwise}.
\end{cases}
\end{equation*}
where $\gamma$ is a scalar such that $\hat{\mb y}^H\mb F \mb F^H \hat{\mb y}=N$ and $\mb F$ is the DFT matrix. In SHAPE algorithm proposed by Rowe \textit{et al.} \cite{Rowe14}, a least-squares fitting approach for the spectral shaping problem for SISO has been formulated by minimizing the following cost function

\begin{equation}
\label{Eq:Stoica}
\|\mb F^H \mb x - \hat{\mb y} \|_2^2
\end{equation}
We extend \eqref{Eq:Stoica} for MIMO radar and employ it as a constraint in the optimization problem as follows.
\be
\|(\mb I_M  \otimes \mb F^H) (\mb 1_M  \otimes \hat{\mb y}  )-\mb x\|_2^2 = \|\bar{\mb F}^H \bar{\mb y}-\mb x\|_2^2 \leq E_R
\label{Eq:SIC}
\ee
where $\mb 1_M=[1, 1,\ldots, 1] \in \mathbb{R}^{M\times 1}$, $\bar{\mb F}=\mb I_M  \otimes \mb F^H$, and $\bar{\mb y}=\mb 1_M  \otimes \hat{\mb y} $, and $E_R$ is the maximum tolerable spectral error.

\subsection{Problem Formulation}
\label{Sec:PF}

The optimization problem can be formulated as the following matching problem:

\be
\label{Eq:BP}
\left\{ \begin{array}{cc}
	\displaystyle
	\min_{\mb x} &\sum_{k=1}^{K}\sum_{p=-\frac{N}{2}}^{\frac{N}{2}-1} [d_{kp}-|\mb a^H_{kp} \mb W_p \mb x|]^2\\
	%        \text{subject to:  } & |\mb x(k)|^2\leq 1/(MN) \\
	\text{s.t.:  } & |x_m(n)|=1, \text{ for }m=1, 2,\ldots, M \text{ and }\\
	 & \quad \quad \quad \quad n=0, 1,\ldots, N-1\\
	& \|\bar{\mb F}^H \bar{\mb y}-\mb x\|_2^2 \le E_R\\
\end{array} \right.
\ee
where $d_{kp} \in \mathbb{R}$ is the desired beampattern. The constraints $|x_m(n)|=1$ represent the constant modulus. These constraints are neither convex nor linear and it is well known in the literature that (\ref{Eq:BP}) is a hard non-convex problem even without the spectral  constraint. He \textit{et al.}  \cite{He11} proposed a solution to problem (\ref{Eq:BP}) without the spectral constraint by employing a peak-to-average ratio constraint as a relaxation of the constant modulus constraint. However, they used the cyclic algorithm \cite{Sussman62,Gerchberg72} to solve the unconstrained problem $\min_{\mb y_p} \sum_{k=1}^{K}\sum_{p=-\frac{N}{2}}^{\frac{N}{2}-1} [d_{kp}-|\mb a^H_{kp} \mb y_p|]^2$ in the first stage and then in the second stage they aim to find the constant modulus approximation of the solution. The algorithm does not directly minimize the cost function under constant modulus constraint or any relaxed version thereof. In this paper, we propose a new solution that minimizes the cost function of interest subject to the contant modulus constraint and the spectral constraint by solving a sequence of problems under a relaxed convex constraint such that constant modulus is still achieved at convergence. The proposed solution has the ability to break the computational cost--solution quality trade-off that has been demonstrated in past work such as SDR with randomization \cite{Luo10,Cui14TSP} or the simulated annealing approach \cite{He11}.

\textit{Remark}: The cost function of \eqref{Eq:BP} can be modified as follows: $\sum_{k=1}^{K}\sum_{p=-\frac{N}{2}}^{\frac{N}{2}-1} w_{kp}[d_{kp}-|\mb a^H_{kp} \mb W_p \mb x|]^2 $ to control the relative importance of certain frequency bands or angles; where $w_{kp}$ are positive weights such that $\sum_{k=1}^{K}\sum_{p=-\frac{N}{2}}^{\frac{N}{2}-1} w_{kp} = 1$. Note such a modification can also be easily accommodated in the analytical development presented next.

\section{Beampattern Design under Constant Modulus and Spectral Constraints}
\label{Sec:BIC}

\subsection{Non-convex Optimization Problem}
As shown in \cite{He11}, it is more convenient to rewrite the objective function of \eqref{Eq:BP} as
\be
\label{Eq:NA}
\sum_{k=1}^{K}\sum_{p=-\frac{N}{2}}^{\frac{N}{2}-1}  |d_{kp}e^{j\phi_{kp}}-\mb a^H_{kp} \mb W_p \mb x|^2
\ee
where $\phi_{kp}=\arg\{\mb a^H_{kp} \mb W_p \mb x\}$. Since $\mb x$ is unknown, $\phi_{kp}$ is also unknown for all $k$ and $p$. In the existing literature \cite{He11,Sussman62,Gerchberg72}, this problem has been resolved by an iterative method. This method first minimizes Eq. \eqref{Eq:NA} w.r.t. $\mb x$ for a fixed values of $\{\phi_{kp}\}$ and then finds the optimal $\{\phi_{kp}\}$ for the fixed $\mb x$ obtained in the previous iteration step. It has been shown that such an iterative method ensures that the cost function is monotonically decreasing and converges to a finite value. Therefore, we focus on solving the following constrained problem for a fixed $\{\phi_{kp}\}$.
\be
\label{Eq:BPO}
\left\{ \begin{array}{cc}
	\displaystyle
	\min_{\mb x} &\sum_{k=1}^{K}\sum_{p=-\frac{N}{2}}^{\frac{N}{2}-1} |d_{kp}e^{j\phi_{kp}}-\mb a^H_{kp} \mb W_p \mb x|^2\\
	%        \text{subject to:  } & |\mb x(k)|^2\leq 1/(MN) \\
	\text{s.t.:  } & |x_m(n)|=1, \text{ for }m=1, 2,\ldots, M \text{ and }\\
	 & \quad \quad \quad \quad n=0, 1,\ldots, N-1\\
	& \|\bar{\mb F}^H \bar{\mb y}-\mb x\|_2^2 \le E_R\\
\end{array} \right.
\ee
First, let us define the following
\be
\mb A_{p}={\begin{bmatrix}
		\mb a^H_{1p} \\
		\vdots\\
		\mb a^H_{Kp}
\end{bmatrix}}, \quad
\mb d_{p}={\begin{bmatrix}
		d_{1p} e^{j\phi_{1p}}\\
		\vdots\\
		d_{Kp} e^{j\phi_{Kp}}
\end{bmatrix}}
\ee
Then the objective function of \eqref{Eq:BPO} can be rewritten in terms of $\mb A_{p}$ and $\mb d_{p}$ \cite{Aldayel17TSP}
\begin{align}
\label{Eq:NCx}
f(\mb x)=&\sum_{p} \|\mb d_p-\mb A_p \mb W_p \mb x \|_2^2\\
% =&\sum_{p} \mb x^H \mb W^H_p \mb A^H_p\mb A_p \mb W_p\mb x - \mb d^H_p \mb A_p \mb W_p\mb x - \mb x^H \mb W^H_p \mb A^H_p \mb d_p\nonumber\\
% &+\sum_{p}\mb d_p^H\mb d_p\nonumber\\
% =& \mb x^H \big(\sum_{p}  \mb W^H_p \mb A^H_p\mb A_p \mb W_p\big)\mb x - \big(\sum_{p}\mb d^H_p \mb A_p \mb W_p\big)\mb x \nonumber\\
%  &- \mb x^H \big (\sum_{p}\mb W^H_p \mb A^H_p \mb d_p \big)+\sum_{p}\mb d_p^H\mb d_p\nonumber\\
=&\mb x^H \mb P \mb x - \mb q^H \mb x - \mb x^H \mb q+r
\end{align}
where $\mb P=\sum_{p}  \mb W^H_p \mb A^H_p\mb A_p \mb W_p$, $\mb q=\sum_{p}\mb W^H_p \mb A^H_p \mb d_p$ and $r=\sum_{p}\mb d_p^H\mb d_p$. Moreover, the spectral constraint can also be simplified as
\medskip
\begin{align}
\|\bar{\mb F}^H \bar{\mb y}-\mb x\|^2_2 &=(\bar{\mb F}^H \bar{\mb y}-\mb x)^H (\bar{\mb F}^H \bar{\mb y}-\mb x)\nonumber\\
&=\mb x^H \mb x -2\re\{ \bar{\mb y}^H \bar{\mb F}\mb x\}+\bar{\mb y}^H \bar{\mb F}\bar{\mb F}^H \bar{\mb y}\nonumber\\
&=2L-2\re\{\bar{\mb y}^H \bar{\mb F}\mb x\}\nonumber
\end{align}
where $L=MN$. Hence, the spectral constraint can be rewritten as
\be
\re\{\bar{\mb y}^H \bar{\mb F}\mb x\} \geq(1-E_R/2)L\nonumber
\ee

The optimization problem \eqref{Eq:BPO} is equivalent to the following problem.
\be
\label{Eq:BPMatrixform}
\left\{ \begin{array}{cc}
	\displaystyle
	\min_{\mb x} & \mb x^H \mb P \mb x - \mb q^H \mb x - \mb x^H \mb q+r\\
	\text{s.t.:  } & |x_m(n)|=1, \text{ for }m=1, 2,\ldots, M \text{ and }\\
	 & \quad \quad \quad \quad n=0, 1,\ldots, N-1\\
	& \re\{\bar{\mb y}^H \bar{\mb F}\mb x\} \geq(1-E_R/2)L\\
	% & \|\mb x-\mb x_0\|\le E_R
\end{array} \right.
\ee
Moreover, $f(\mb x)$ can be converted to the following function with \emph{real} (as opposed to complex) variables.
\be
f_R(\mb u)=\mb u^T \mb G \mb u -  \mb t^T\mb u-\mb u^T\mb t+r
\ee
where
\be
\mb u = [\re \{\mb x\}^T \im \{\mb x\}^T]^T
\ee
\be
\mb G={\begin{bmatrix}
		\re \{\mb P\}   &-\im \{\mb P\}\\
		\im \{\mb P\}   &\re \{\mb P\}\\
\end{bmatrix}}
\ee
\be
\mb t={\begin{bmatrix}
		\re \{\mb q\} \\
		\im \{\mb q\}  \\
\end{bmatrix}}
\ee
The problem \eqref{Eq:BPMatrixform} can be rewritten as
\be
\label{Eq:BPreal}
\left\{ \begin{array}{cc}
	\displaystyle
	\min_{\mb s} &\mb s^T (\mb R+\lambda \mb I) \mb s\\
	\text{s.t.:  } & \mb s^T \mb E_l \mb s =1, \quad l=1,2,\ldots,L\\
	&  \bar{\mb s}^T\mb s\geq(1-E_R/2)L
\end{array} \right.
\ee
where $\lambda$ is an arbitrary positive number,
\be
\label{Eq:sbar}
\mb {\bar s}= [\re \{\bar{\mb F}^H\bar{\mb y}\}^T \; \im \{\bar{\mb F}^H\bar{\mb y}\}^T \; 0]^T,
\ee
\be
\mb R={\begin{bmatrix}
		\mb G   &-\mb t\\
		-\mb t^T & r\\
\end{bmatrix}},
\ee
\be
\mb s={\begin{bmatrix}
		\re \{\mb x\} \\
		\im \{\mb x\}  \\
		1\\
\end{bmatrix}},
\ee
% \een
% \ben
% \mb s={\begin{bmatrix}
%             \mb u\\
%              1\\
%          \end{bmatrix}},
% \een
and $\mb E_{l}$ is a $2L+1\times 2L+1$ matrix given by
\be
\mb E_{l}(i,j)=
\begin{cases}
1 & \text{if } i=j= l\\%, \text{ and } l \le L,\\
1 & \text{if } i=j= l+L\\%, \text{ and } l \le L,\\
%1 & \text{if } i=j=2L+1, \text{ and } l = L+1,\\
0 & \text{otherwise}.
\end{cases}
\ee
Note that, since
\begin{align}
\mb s^T \mb R \mb s & = \mb x^H \mb P \mb x - \mb q^H \mb x - \mb x^H \mb q+r\\
& = \sum_{p} \|\mb d_p-\mb A_p \mb W_p \mb x \|_2^2\\
& \ge 0
\end{align}
, $\mb R$ is positive semi-definite. Further, because the problem \eqref{Eq:BPreal} enforces constant modulus, i.e., $\mb s^T \mb E_l \mb s =1$ for $l=1,2,\ldots,L$, $\lambda \mb s^T \mb s$ is a constant value ($\lambda \mb s^T \mb s=\lambda (L+1)$). As a result, \eqref{Eq:BP} and \eqref{Eq:BPreal} are the identical optimization problems and the optimal solution of \eqref{Eq:BP} and the resulting complex solution of \eqref{Eq:BPreal} are also identical for any $\lambda \geq 0$.
\subsection{Sequence of Closed Form Solutions}
\label{Sec:SCF}
Now we focus on solving \eqref{Eq:BPreal}. Though it is minimization of a convex objective function, it is still non-convex because of the constant modulus constraint. We propose a new sequential approach to solve \eqref{Eq:BPreal} which involves solving a sequence of convex problems. Let us consider the following sequence of constrained QPs where the $n$-th QP is given by
\be
\label{Eq:BPCP}
(CP)^{(n)}
\left\{ \begin{array}{cc}
	\displaystyle
	\min_{\mb s} & \mb s^T (\mb R+\lambda \mb I) \mb s\\
	\text{s.t.:  } & \mb B^{(n)}\mb s= \mb 1\\
	&\bar{\mb s}^{(n)T}\mb s\geq(1-E_R/2)L
\end{array} \right.
\ee
where $\bar{\mb s}^{(n)}$ is given by: 
% \odor e^{j \boldsymbol{\beta}}

\be
\label{Eq:sbarn}
\bar{\mb s}^{(n)}={\begin{bmatrix}
		\re \{(\bar{\mb F}^H\bar{\mb y} )\odot e^{\{j \arg(\mb x^{(n-1)})-\arg(\bar{\mb F}^H \bar{\mb y})\} } \} \\
		\im \{(\bar{\mb F}^H\bar{\mb y} )\odot e^{\{j \arg(\mb x^{(n-1)})-\arg(\bar{\mb F}^H \bar{\mb y})\} }\}\\
		0
\end{bmatrix}}
\ee
% \be
% \label{Eq:sbarn}
% \bar{\mb s}^{(n)}= [\re \{(\bar{\mb F}^H(\bar{\mb y} )\odot e^{j \arg(\mb x^{(n-1)})-\arg(\bar{\mb F}^H \bar{\mb y})} \}^T \; \im \{\bar{\mb F}^H\bar{\mb y}\}^T \; 0]^T,
% \ee
and $\mb B^{(n)}=[\mb b^{(n)}_{1}, \mb b^{(n)}_{2} , ..., \mb b^{(n)}_{L+1}]^T \in \mathbb{R}^{(L+1)\times (2L+1)}$ such that the line defined by ${\mb b^{(n)T}_{l}}\mb s=1$ is a tangent to the circle $\mb s^T \mb E_l \mb s =1$ for $l=1, 2,\ldots, L$. Specifically, $\mb b_l$ is given by
\begin{align}
\label{Eq:bnn}
\mb b_l^{(n)}(i)= \begin{cases}
\cos(\gamma_l^{(n)}) &\text{if } i=l\\
\sin(\gamma_l^{(n)}) &\text{if } i=l+L\\
0 & \text{otherwise}.
\end{cases}
\end{align}
for $l=1,\ldots,L$ and $\mb b_{L+1}^{(n)} = [0,\ldots,0,1]^T$ where $\gamma_l^{(n)} = 2 \arg(x_l^{(n-1)}) - \gamma_l^{(n-1)}$ and $x_l^{(n)}$ is the $l$-th elements of $\mb x^{(n)}$ which is the complex version of the optimal solution of \eqref{Eq:BPCP}, $\mb s^{(n)}$, that is, $x_l^{(n)}= s_l^{(n)}+j s_{l+L}^{(n)}$ and conversely $\mb s^{(n)} = [\re \{\mb x^{(n)}\}^T \im \{\mb x^{(n)}\}^T \,\, 1]^T$. Note that, the term $e^{\{j \arg(\mb x^{(n-1)})-\arg(\bar{\mb F}^H \bar{\mb y})\} } \}$ in (\ref{Eq:sbarn}) depends on the argument $\mb x^{(n-1)}$, which changes $\bar{\mb s}^{(n)}$  in each iteration.

Although the problem \eqref{Eq:BPCP} does not result in a constant modulus solution, a sequence of such problems (in the index $n$) ensures a non-increasing sequence of cost function values, such that the sequence of the corresponding optimal solutions converges to constant modulus for large enough $\lambda$\footnote{For a formal proof of this, see \cite{BIC_TR_2017}}. To recognize this, we first show that the constraints of $CP^{(n)}$ in \eqref{Eq:BPCP} are adjusted so that the feasible set of $CP^{(n)}$ includes $\mb x^{(n-1)}$.

\begin{lemma}
\label{Lem:FeasibleSet}
The feasible set of problem $CP^{(n)}$ contains the optimal solution of problem $CP^{(n-1)}$.
\end{lemma}

\begin{proof}
Let $\mb s^{(n-1)}$ be the optimal solution of $CP^{(n-1)}$. Then $\mb B^{(n-1)}\mb s^{(n-1)}=\mb 1$ and $\bar{\mb s}^{(n-1)T}\mb s^{(n-1)}\geq(1-E_R/2)L$. Let $x_l^{(n-1)}=\rho_l e^{j\psi_l}$, then $(\mb B^{(n-1)}\mb s^{(n-1)})_l$, the $l$-th element of $\mb B^{(n-1)}\mb s^{(n-1)}$, should be equal to 1. That is,
\begin{align}
\label{Eq:Bsl}
(\mb B^{(n-1)}\mb s^{(n-1)})_l=&\re\{x_l^{(n-1)}\} \cos(\gamma_l^{(n-1)})+\nonumber\\
&\im\{x_l^{(n-1)}\}\sin(\gamma_l^{(n-1)})\\
=&\rho_l \cos(\psi_l) \cos(\gamma_l^{(n-1)})+\nonumber\\
&\rho_l \sin(\psi_l) \sin(\gamma_l^{(n-1)})\\
=&1
\end{align}
where $\gamma_l^{(n)} = 2 \arg(x_l^{(n-1)}) - \gamma_l^{(n-1)}$. This implies
\be
\rho_l=\frac{1}{\cos(\psi_l) \cos(\gamma_l^{(n-1)})+\sin(\psi_l) \sin(\gamma_l^{(n-1)})}
\ee
Note that $\mb s^{(n-1)}$ belongs to the feasible set of $CP^{(n)}$ if and only if $\mb B^{(n)}\mb s^{(n-1)}=\mb 1$ and $\bar{\mb s}^{(n)T}\mb s^{(n-1)}\geq(1-E_R/2)L$. We have
\begin{align}
% \label{eq:Bsnl}
(\mb B^{(n)}\mb s^{(n-1)})_l=&\rho_l \cos(\psi_l) \cos(\gamma_l^{(n)})\nonumber\\
&+\rho_l \sin(\psi_l) \sin(\gamma_l^{(n)})\\
=&\rho_l \cos(\psi_l - \gamma_l^{(n)}) \label{Eq:gammal}\\
=&\rho_l \cos(\psi_l - 2\psi_l + \gamma_l^{(n-1)}) \\
=&\rho_l \cos(\psi_l - \gamma_l^{(n-1)})\\
=&\rho_l \cos(\psi_l) \cos(\gamma_l^{(n-1)})\nonumber\\
&+ \rho_l \sin(\psi_l) \sin(\gamma_l^{(n-1)})\\
=&1\label{Eq:Proof1}
\end{align}
Note that we used $\gamma_l^{(n)} = 2 \arg(x_l^{(n-1)}) - \gamma_l^{(n-1)}=2\psi_l- \gamma_l^{(n-1)}$ in (\ref{Eq:gammal}). To show $\bar{\mb s}^{(n)T}\mb s^{(n-1)}\geq(1-E_R/2)L$, let $\bar{\mb x}$ denote the complex version of $\mb {\bar s}$, that is, $\mb {\bar s}= [\re \{\bar{\mb x}\}^T \im \{\bar{\mb x}\}^T]^T$. Then we have
\begin{align}
\label{eq:sl}
(1-E_R/2)L&\le \bar{\mb s}^{(n-1)T}\mb s^{(n-1)}\\
&=\re\{\bar{\mb x}^{(n-1)H}\mb x^{(n-1)}\}\\
&=\re\{\sum_l^L \bar{x}_l^{*(n-1)} \rho_l e^{j\psi_l}\}\\
&\leq \Big|\sum_l^L \bar{x}_l^{*(n-1)} \rho_l e^{j\psi_l}\Big|\\
&\leq \sum_l^L \Big|\bar{x}_l^{*(n-1)} \rho_l e^{j\psi_l}\Big|\\
&\leq \sum_l^L \Big|\bar{x}_l^{*(n-1)}\Big| \rho_l\\
&= \sum_l^L  |\bar{x}_l^{*(n-1)}| e^{-j\psi_l} \rho_l e^{j\psi_l}\label{Eq:Lemma1.1}\\
&= \sum_l^L  \bar{x}_l^{*(n)} \rho_l e^{j\psi_l}\label{Eq:Lemma1.2}\\
&= \re\{\bar{\mb x}^{(n)H}\mb x^{(n-1)}\}\\
&= \bar{\mb s}^{(n)T}\mb s^{(n-1)}\label{Eq:Proof2}
\end{align}
Note that the equality between \eqref{Eq:Lemma1.1} and \eqref{Eq:Lemma1.2} holds because we define $\bar{\mb s}^{(n)}$ such that $\arg(\bar{\mb F}^H \bar{\mb y})=\arg(\mb x^{(n-1)})$. Eqs. \eqref{Eq:Proof1} and \eqref{Eq:Proof2} confirm $\mb B^{(n)}\mb s^{(n-1)}=\mb 1$ and $\bar{\mb s}^{(n)T}\mb s^{(n-1)}\geq(1-E_R/2)L$.
\end{proof}

Lemma \ref{Lem:FeasibleSet} proves that the feasible set of each iteration is updated such that it contains the optimal solution of the optimization problem at the previous iteration step. If $|\mb x^{(n)}|= \mb 1$, then the constraints of the next problem $CP^{(n+1)}$ are the same as problem $CP^{(n)}$, which means $\mb x^{(n+1)}= \mb x^{(n)}$ and, hence, the algorithm converges. Lemma \ref{Lem:Converge} further establishes that the cost function sequence is in fact non-increasing and converges. This procedure is visually illustrated in Fig. \ref{Fig:FCpn}.

\begin{figure}[t]
	\centering
	\subfloat[The initial problem $CP^{(1)}$, the initial feasible set is the blue line.]{\includegraphics[scale=0.32]{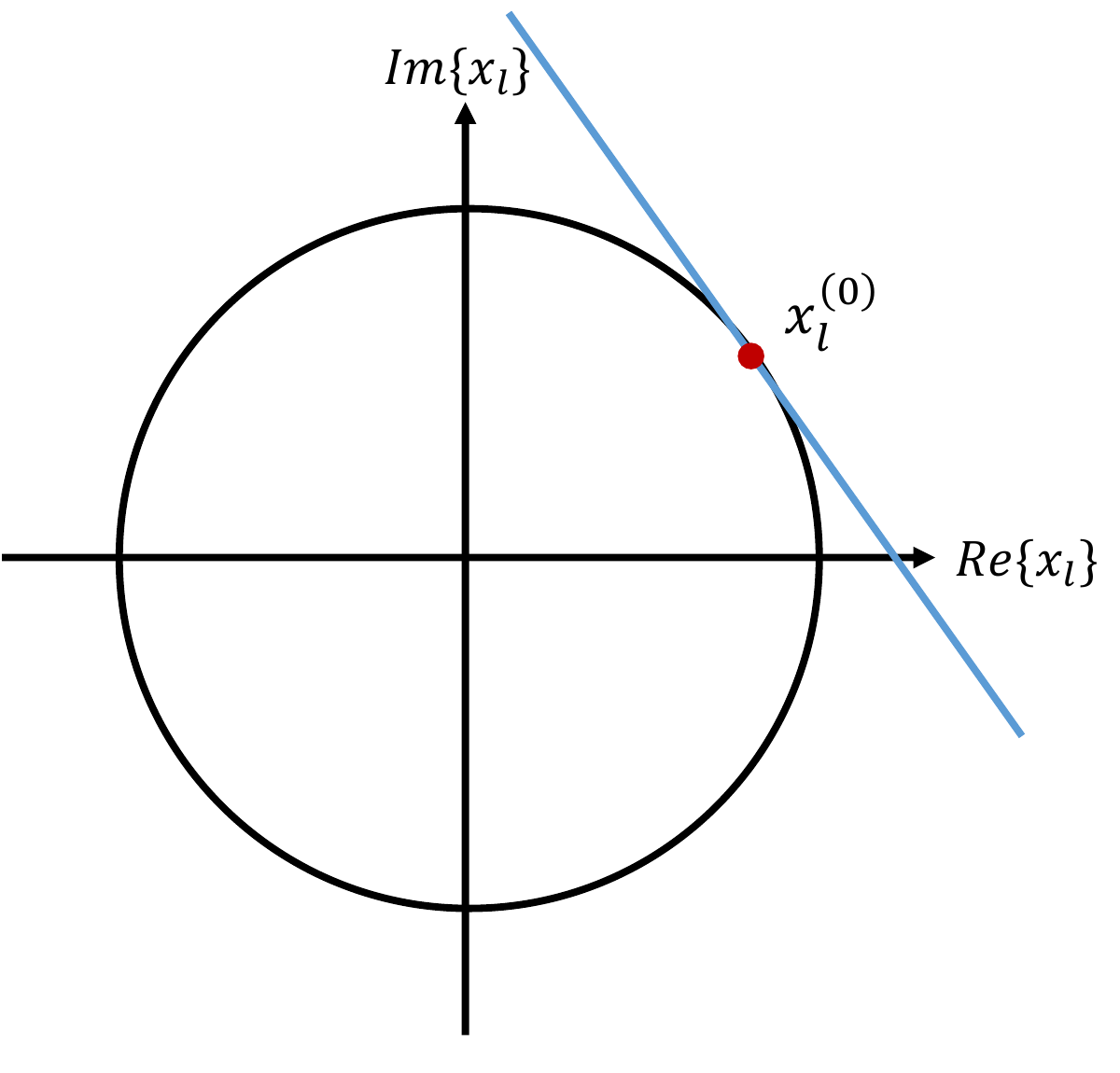}}\quad
	\subfloat[Solution of problem $CP^{(1)}$ lies on the initial feasible set.]{\includegraphics[scale=0.32]{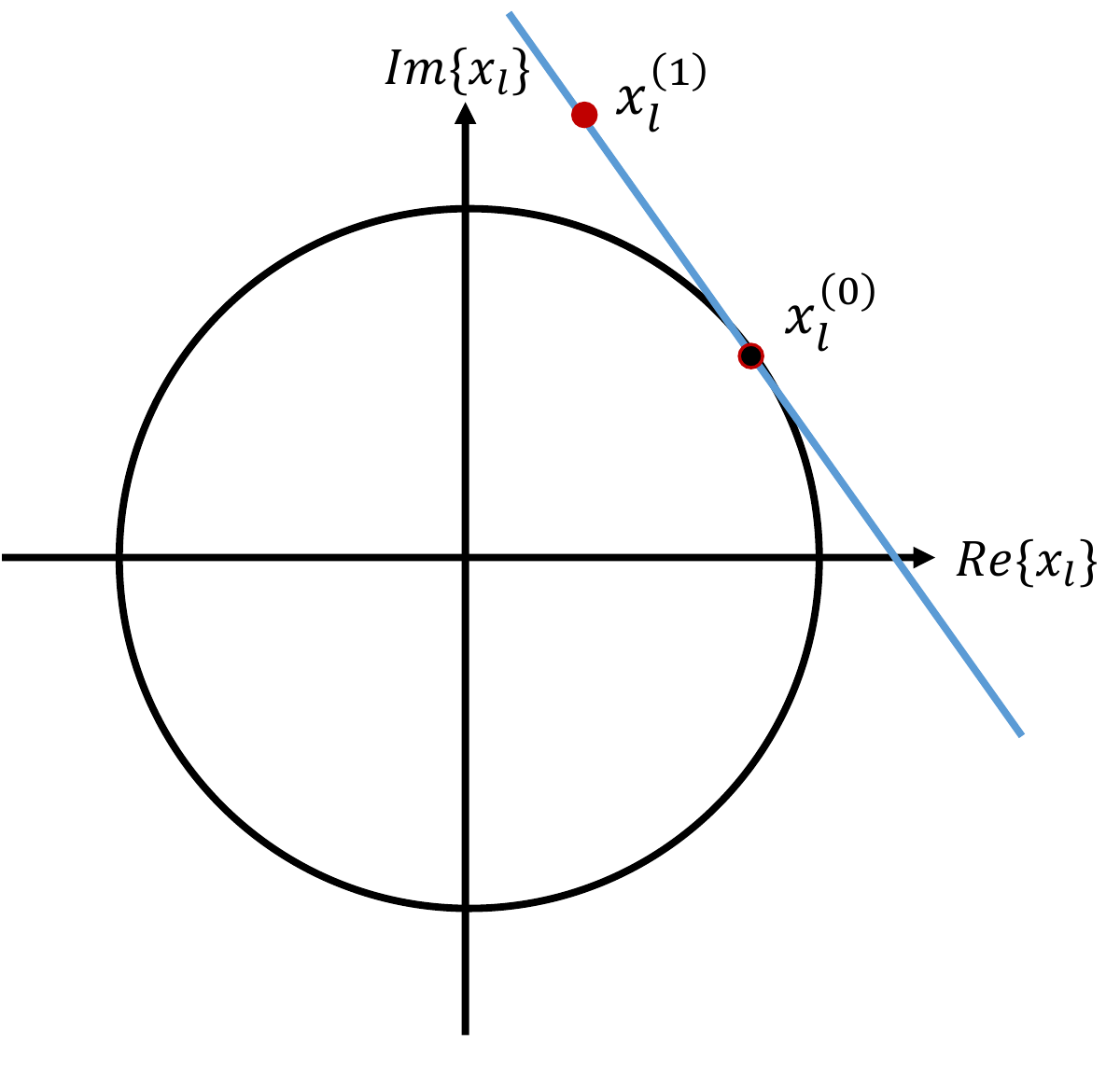}}\\
	\subfloat[The new adjusted feasible set (Contains $x_l^{(1)}$) in blue, the previous feasible set in gray.]{\includegraphics[scale=0.32]{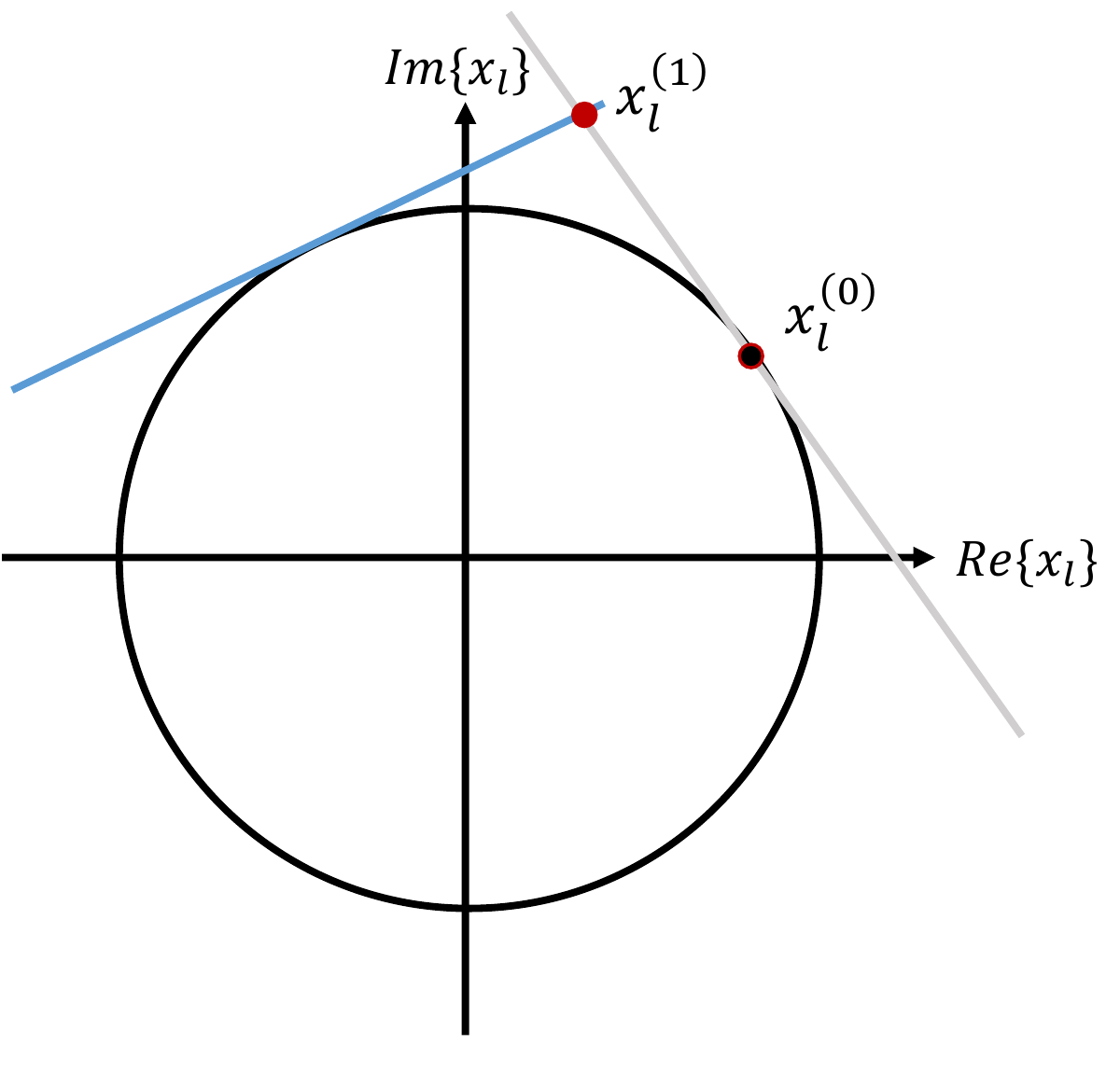}}\quad
	\subfloat[The converged solution now lies on the constant modulus.]{\includegraphics[scale=0.32]{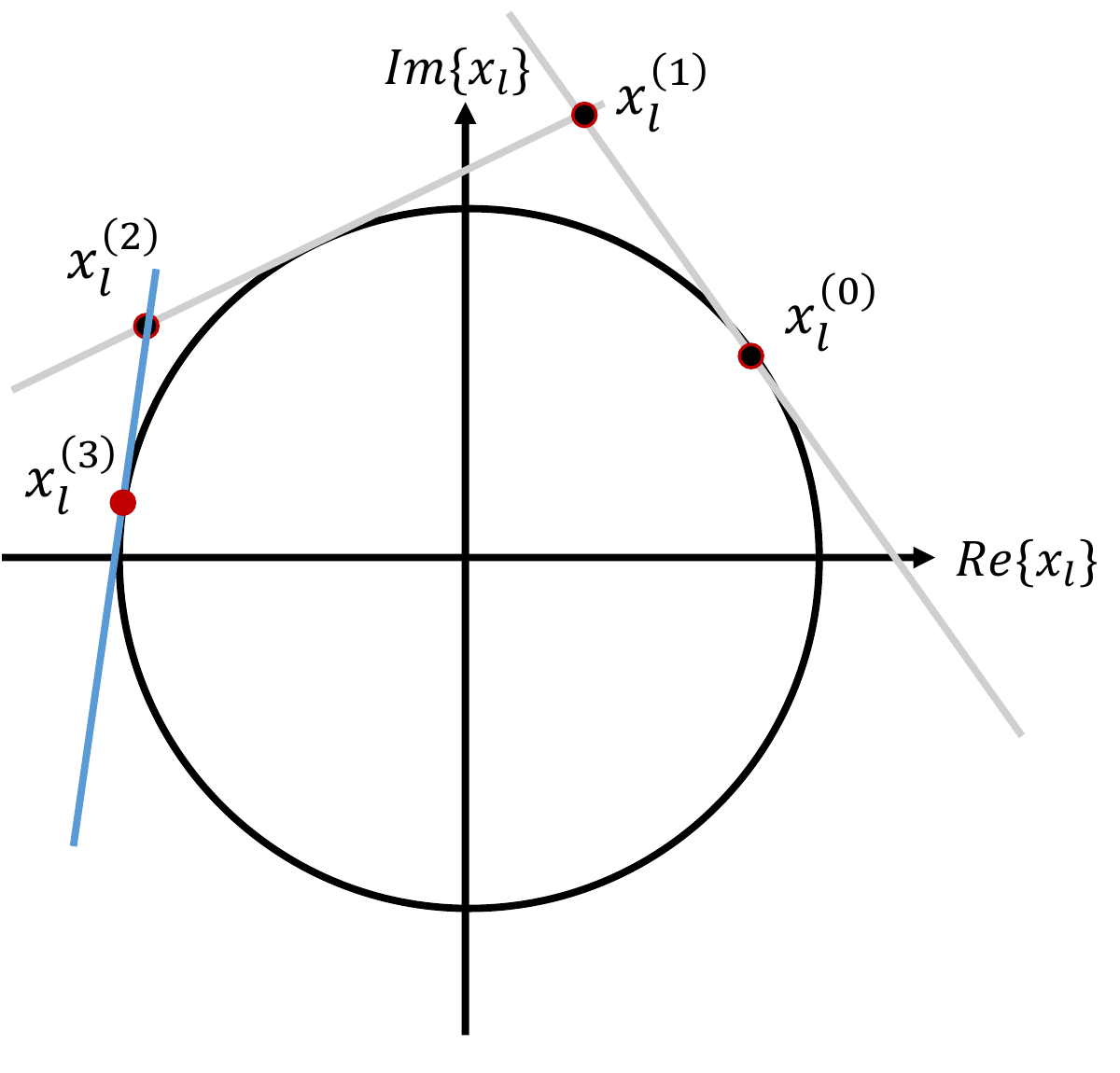}}
	\caption{Illustration of the successive solutions of \eqref{Eq:BPCP} for the $l$-th element of the vector $\mb x^{(n)}$ i.e. $x_l^{(n)}$. The current feasible set is shown via a blue line.\label{Fig:FCpn}}
\end{figure}

Now we focus on how to solve the optimization problem \eqref{Eq:BPCP} at each iteration step. Note that the problem \eqref{Eq:BPCP} is a convex quadratic minimization with linear equality constraints. Using the optimality conditions for problem \eqref{Eq:BPCP}, the sufficient and necessary Karush-Kuhn-Tucker (KKT) conditions \cite{Boyd04} of \eqref{Eq:BPCP} give the following.
\be
\label{Eq:KKT1}
2 (\mb R+\lambda \mb I)  \mb s^{(n)}+\mb B^{(n)T}\mb v^{(n)}-\mu^{(n)}\bar{\mb s}=0
\ee
\be
\label{Eq:KKT1b}
\mb B^{(n)}\mb s^{(n)}= \mb 1
\ee
\be
\label{Eq:KKT2}
\mu^{(n)}\big(\bar{\mb s}^{(n)T}\mb s^{(n)}-(1-E_R/2)L\big)=0
\ee
\be
\label{Eq:KKTic}
\bar{\mb s}^{(n)T}\mb s^{(n)}-(1-E_R/2)L\ge0
\ee
\be
\label{Eq:KKTmu}
\mu^{(n)}\ge0
\ee
We can directly solve these equations to find $\mb s^{(n)}$, $\mb v^{(n)}$ and $\mu^{(n)}$. The complementary slackness condition \eqref{Eq:KKT2} implies that either $\mu^{(n)}=0$ or $\bar{\mb s}^{(n)T}\mb s^{(n)} -(1-E_R/2)L=0$ must be satisfied. In the case of $\mu^{(n)}=0$, from Eqs. \eqref{Eq:KKT1} and \eqref{Eq:KKT1b}, we have
\be
\label{Eq:SCMC}
\begin{bmatrix}
				\bar{\mb R} &  {\mb B^{(n)}}^{T} \\
				\mb B^{(n)} & \mb 0
\end{bmatrix}
\begin{bmatrix}
		\mb s^{(n)}\\
		\mb v^{(n)}
\end{bmatrix}=
\begin{bmatrix}
		\mb 0\\
		\mb 1
\end{bmatrix}
\ee
where $\bar{\mb R}=2(\mb R+\lambda \mb I)$ and $\mb v^{(n)}\in \mathbb{R}^{(L+1)\times 1}$ is the Lagrange multiplier associated with the equality constraints. Solving (\ref{Eq:SCMC}) by block elimination gives
\be
\label{Eq:sl}
\hat{\mb s}^{(n)}=\bar{\mb R}^{-1}{\mb B^{(n)}}^{T}\big(\mb B^{(n)} \bar{\mb R}^{-1} {\mb B^{(n)}}^{T}\big)^{-1}\mb 1
\ee
If $\hat{\mb s}^{(n)}$ satisfies $\bar{\mb s}^{(n)T}\hat{\mb s}^{(n)}-(1-E_R/2)L\ge 0$, then $\mb s^{(n)}=\hat{\mb s}^{(n)}$ is the optimal solution of problem ($CP^{(n)}$). However, if $\bar{\mb s}^{(n)T}\hat{\mb s}^{(n)} -(1-E_R/2)L<0$, then $\hat{\mb s}^{(n)}$ is not the solution since it violates \eqref{Eq:KKTic}. Thus, $\mu^{(n)}=0$ can not be valid and, therefore, it is the case that $\bar{\mb s}^{(n)T}\mb s^{(n)} -(1-E_R/2)L=0$ must holds. In this case, the KKT conditions \eqref{Eq:KKT1} through \eqref{Eq:KKT2} are given in the matrix form by
\be
\label{Eq:KKTss}
{\begin{bmatrix}
		\bar{\mb R} & {\mb B^{(n)}}^{T} & -\bar{\mb s}^{(n)}\\
		\mb B^{(n)} & \mb 0 & \mb 0\\
		-\bar{\mb s}^{(n)T} & \mb 0 & \mb 0
\end{bmatrix}}
{\begin{bmatrix}
		{\mb s}^{(n)}\\
		{\mb v}^{(n)}\\
		{\mu}^{(n)}
\end{bmatrix}}=
{\begin{bmatrix}
		\mb 0\\
		\mb 1\\
		-(1-E_R/2)L
\end{bmatrix}}
\ee
Using block elimination to solve \eqref{Eq:KKTss} gives
\be
\label{Eq:sl2}
\mb s^{(n)}={\mu}^{(n)}\bar{\mb R}^{-1}(\mb I-{\mb B^{(n)}}^{T}\hat{\mb R}\mb B^{(n)} \bar{\mb R}^{-1})\bar{\mb s}^{(n)}+\hat{\mb s}^{(n)}
\ee
where
\be
\hat{\mb R}=\big(\mb B^{(n)} \bar{\mb R}^{-1} {\mb B^{(n)}}^{T}\big)^{-1}
\ee
\be
\label{Eq:KKTmu2}
{\mu}^{(n)}=\frac{1}{\alpha^{(n)}}\big(\bar{\mb s}^{(n)T}\hat{\mb s}^{(n)}-(1-E_R/2)L\big)
\ee
\be
\label{Eq:Alpha}
\alpha^{(n)}=-{\begin{bmatrix}
		\bar{\mb s}^{(n)}\\
		\mb 0
	\end{bmatrix}^T
	\begin{bmatrix}
		\bar{\mb R} & {\mb B^{(n)}}^{T}\\
		\mb B^{(n)} & \mb 0
\end{bmatrix}}^{-1}
{\begin{bmatrix}
		\bar{\mb s}^{(n)}\\
		\mb 0
\end{bmatrix}}
\ee
Note that \eqref{Eq:KKTic} always holds since $\bar{\mb s}^T\mb s^{(n)} -(1-E_R/2)L=0$ in this case. To confirm all KKT conditions are satisfied, we have to show the dual feasibitity condition \eqref{Eq:KKTmu} holds. The following lemma proves this.

\begin{lemma}
	\label{Lemma:ps}
	If $\bar{\mb s}^T\hat{\mb s}^{(n)} -(1-E_R/2)L<0$ then ${\mu}^{(n)}> 0$.
\end{lemma}

\begin{proof}
First, let
\be
\mb K=
{\begin{bmatrix}
		\bar{\mb R} & {\mb B^{(n)}}^{T} & -\bar{\mb s}^{(n)}\\
		\mb B^{(n)} & \mb 0 & \mb 0\\
		-\bar{\mb s}^{(n)T} & \mb 0 & \mb 0
\end{bmatrix}}
\ee
\be
\mb K_{11}=
{\begin{bmatrix}
		\bar{\mb R} & {\mb B^{(n)}}^{T}\\
		\mb B^{(n)} & \mb 0
\end{bmatrix}}
\ee
If $\bar{\mb s}^{(n)}$ is linearly dependent on $\mb b^{(n)}_{1}, \mb b^{(n)}_{2} , \ldots, \mb b^{(n)}_{L+1}$ and $\bar{\mb s}^T\hat{\mb s}^{(n)} -(1-E_R/2)L<0$, then there will be no solution to $CP^{(n)}$ which contradicts Lemma \ref{Lem:FeasibleSet}. Therefore, $\mb b^{(n)}_{1}, \mb b^{(n)}_{2}, \ldots, \mb b^{(n)}_{L+1}$, and $\bar{\mb s}$ must be linearly independent. Moreover, since $\bar{\mb R}$ is positive definite, all the eigenvalues of $\mb K$ are nonzero according to Theorem 2.1 in \cite{Higham98}, which means $\mb K$ is nonsingular. Since $\mb K$ is nonsingular, the Schur complement of the block $\mb K_{11}$ in $\mb K$ is also nonsingular (nonzero in our case) according to Section C.4 in \cite{Boyd04} and equals to $\alpha^{(n)}$. This implies
\be
\label{Eq:alphanz}
\alpha^{(n)}\ne 0
\ee
Using the block inverse to the matrix $\mb K_{11}$, Eq. \eqref{Eq:Alpha} can be rewritten as
\begin{align}
\alpha^{(n)} &=-\bar{\mb s}^{(n)T}(\bar{\mb R}^{-1}-\bar{\mb R}^{-1}{\mb B^{(n)}}^{T}\hat{\mb R}\mb B^{(n)} \bar{\mb R}^{-1})\bar{\mb s}^{(n)}\\
&=-\bar{\mb s}^{(n)T}\bar{\mb R}^{-\frac{1}{2}}(\mb I-\bar{\mb R}^{-\frac{1}{2}}{\mb B^{(n)}}^{T}\hat{\mb R}\mb B^{(n)} \bar{\mb R}^{-\frac{1}{2}})\bar{\mb R}^{-\frac{1}{2}}\bar{\mb s}^{(n)}\\
&=-\mb y^T(\mb I-\bar{\mb R}^{-\frac{1}{2}}{\mb B^{(n)}}^{T}\big(\mb B^{(n)} \bar{\mb R}^{-1} {\mb B^{(n)}}^{T}\big)^{-1}\mb B^{(n)} \bar{\mb R}^{-\frac{1}{2}})\mb y\\
&=-\mb y^T(\mb I-\mb C(\mb C^T \mb C)^{-1}\mb C^T)\mb y
\end{align}where $\mb y=\bar{\mb R}^{-\frac{1}{2}}\bar{\mb s}^{(n)}$ and $\mb C=\bar{\mb R}^{-\frac{1}{2}}{\mb B^{(n)}}^{T}$. Note that $\mb C_p=\mb C(\mb C^T \mb C)^{-1}\mb C^T$ is an idempotent matrix with eigenvalues of either $0$ or $1$ \cite{Horn12}. This implies that $(\mb I-\mb C_p)$ is positive semidefinite. Therefore,
\be
\label{Eq:alphap}
\alpha^{(n)}\le 0
\ee
Combining (\ref{Eq:alphanz}) and (\ref{Eq:alphap}) implies that $\alpha^{(n)}<0$ and, hence,  ${\mu}^{(n)}> 0$. 
\end{proof}

The process of solving \eqref{Eq:BPMatrixform} for fixed $\{\phi_{kp}\}$ is given in Algorithm \ref{Alg:BIC}. Note that both cases lead to the closed form solutions. The complete BIC algorithm to solve \eqref{Eq:BPO} (including iteration of $\mb x$ and $\{\phi_{kp}\}$) is given in Algorithm \ref{Alg:BIC2}. 
%The structure of Algorithm \ref{Alg:BIC} bears a high level conceptual similarity to our recent work in beampattern design \cite{Aldayel16Radarcon} since the solution in \cite{Aldayel16Radarcon} also employs a sequence of convex problems approach. However, the way to update the feasible set is different so that the feasible set of the next step should include the optimal solution of the previous step. Furthermore, the BIC can deal with the spectral interference constraint as well while \cite{Aldayel16Radarcon} does not.

\begin{algorithm}[!t]
	\caption{Successive algorithm to solve \eqref{Eq:BPMatrixform}}
	\label{Alg:BIC}
	\begin{algorithmic}
		\State \tb{Inputs: } $\mb d_p$, $\mb W_p$, $\mb a_{kp}$ for $p= -\frac{N}{2},...,0, ..., \frac{N}{2}-1$, $k=1, 2, .., K$ and $\zeta$ (the stopping threshold).
		\State \tb{Output: } A solution $\mb x^{\star}$ for problem \eqref{Eq:BPMatrixform}.
		% \State (1) Solve the QCQP of $RC^{(0)}$ using SOCP.
		\State (1) Set $n=1$ and an initial value for $\mb x^{(0)}$.
		\State (2) Compute $\mb B^{(n)}=[\mb b^{(n)}_{1}, \mb b^{(n)}_{2} , ..., \mb b^{(n)}_{L+1}]^T $ as in (\ref{Eq:bnn}).
		% \\(The computational complexity of this step is $\mathcal{O}({M}^{3.5}{N}^{3.5})$)
		\State (3) Compute $\hat{\mb s}^{(n)}$ via eq. (\ref{Eq:sl}) and $\bar{\mb s}^{(n)}$  via eq. (\ref{Eq:sbarn}).
		\State (4) Check the following:
		\If{ $\bar{\mb s}^{{(n)}T}\hat{\mb s}^{(n)}-(1-E_R/2)L\ge 0$}
		\State $\mb s^{(n)}=\hat{\mb s}^{(n)}$.
		\Else
		\State $\mb s^{(n)}={\mu}^{(n)}\bar{\mb R}^{-1}(\mb I-{\mb B^{(n)}}^{T}\hat{\mb R}\mb B^{(n)} \bar{\mb R}^{-1})\bar{\mb s}^{(n)}+\hat{\mb s}^{(n)}$\\
		where ${\mu}^{(n)}$ is defined in (\ref{Eq:KKTmu2}).
		\EndIf
		
		% \State (5) Set $x_{l(n)}=\exp\{j \arg(s_l^{(n)}+j s_{l+L}^{(n)})\}$, $l=1, 2, ..., L$.
		% \\(The computational complexity of this step is $\mathcal{O}({M}^{2}{N}^{2})$ \cite{golub2012matrix})
		\State (5) Construct $\mb x^{(n)}$ where  $x_l^{(n)}= s_l^{(n)}+j s_{l+L}^{(n)}$ for $l=1,...,L$. Check the following:
		\If{ $\sum_{p} \|\mb d_p-\mb A_p \mb W_p\mb x^{(n)}\|_2^2-\sum_{p} \|\mb d_p-\mb A_p \mb W_p\mb x^{(n-1)}\|_2^2<\zeta$}
		\State STOP.
		\Else
		\State set $n=n+1$ GOTO step (2).
		\EndIf
		\State \tb{Output: } $\mb x^{\star}=\exp\{j \arg(\mb x^{(n)})\}$.
		% \State \tb{Output: } $\mb x^{\star}=$ constant modulus version of $\mb s^{(n)}$
		%\tb{Output: } $\mb x^{\star}=\exp\{j \arg(s_l^{(n)}+j s_{l+L}^{(n)})\}$, $l=1, 2, ..., L$
		
	\end{algorithmic}
\end{algorithm}

\noindent \tb{Computational Complexity:} Based on the computational cost of solving (\ref{Eq:KKTss}) in each iteration, the overall computational complexity of BIC is $\mathcal{O}(F{L}^{2.373})-\mathcal{O}(F{L}^{3})$ \cite{Williams12} where $F$ is the total number of iterations.% In comparison, SDR with randomization has a computational complexity of $\mathcal{O}({L}^{3.5})+\mathcal{O}(T{L}^{2})$ \cite{Aubry12IET} where $T$ is the number of randomization trails. It invariably needs a large number of randomization trials $T>>L$ \cite{Cui14TSP} which makes the term  $\mathcal{O}(T{L}^{2})$ very significant. Therefore, the BIC algorithm has much lower complexity.

% \ti{Remark:} Let $\mb x^{(n-1)}$ be the complex version of the optimal solution of $CP^{(n-1)}$,  i.e. $\mb s^{(n-1)} = [\re (\mb x^{(n-1)T}) \im (\mb x^{(n-1)T})]^T$. The affine constraints of $CP^{(n)}$ are adjusted so that the feasible set of $CP^{(n)}$ includes what we call as the {\em constant modulus version} of $\mb x^{(n-1)}$ given by $\mb x_{(n-1)}=\exp(j\arg(\mb x^{(n-1)}))/\sqrt{MN}$.  If $\mb x^{(n)}= \mb x_{(n-1)}$, then the constraints of the next problem $CP^{(n+1)}$ are the same as problem $CP^{(n)}$ which means $\mb x^{(n+1)}= \mb x^{(n)}$ and, hence, the algorithm converges. Otherwise, the feasible set is adapted to include the constant modulus versions of $\mb x^{(n)}$. Convergence is then guaranteed by Lemma \ref{Lemma:ni} which establishes that the SINR sequence that results by using the constant modulus version of the solution at each iteration, is in fact non-increasing and converges.

\noindent \textbf{Convergence Analysis:} The value of the objective function of the problem \eqref{Eq:BPMatrixform} as a function of the  $\mb x^{(n)}$, i.e.\ the optimal solution of the QP at iteration $n$, is non-increasing in $n$. This is proven next.
\begin{lemma}
	\label{Lem:Converge}
	Define $g(\mb s) = \mb s^T (\mb R+\lambda \mb I) \mb s$. Then
	\be
	g(\mb s^{(n-1)}) \ge g(\mb s^{(n)})
	\ee
	In other words, the sequence $\{g(\mb s^{(n)})\}_{n=0}^{\infty}$ is non-increasing. Moreover, the sequence $\{g(\mb s^{(n)})\}_{n=0}^{\infty}$ converges to a finite value $g^{\star}$.
\end{lemma}

\begin{proof}
Denote the feasible sets of $CP^{(n-1)}$ and $CP^{(n)}$ by $\mathcal{F}_{n-1}$ and $\mathcal{F}_{n}$, respectively. From Lemma \ref{Lem:FeasibleSet}, $\mb s^{(n-1)} \in \mathcal{F}_{n}$. Since $CP^{(n)}$ is a convex problem and $\mb s^{(n)}$ is the optimal solution of $CP^{(n)}$,
\be
\label{eq:seqProb}
\mb s^{(n-1)T} (\mb R+\lambda \mb I) \mb s^{(n-1)} \ge \mb s^{(n)T} (\mb R+\lambda \mb I) \mb s^{(n)}
\ee
Therefore, the sequence $\{g(\mb s^{(n)})\}_{n=0}^{\infty}$ is non-increasing. Since $g(\mb s)\ge 0$ for all values of $\mb s$, it is bounded below. Hence, it converges to a finite value $s^{\star}$ according to the monotone convergence theorem \cite{Royden10}.
\end{proof}

\begin{figure}
	\centering
	\includegraphics[scale=0.33]{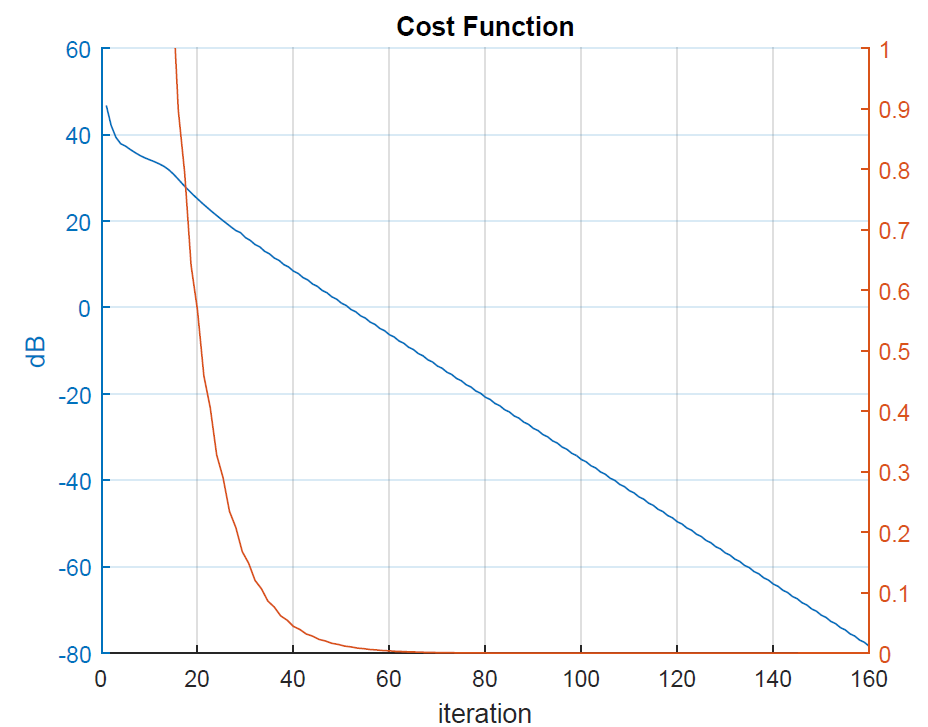}
	\caption{Value of cost function vs. iteration (red curve for the linear scale and blue curve for the log scale).\label{Fig:Nonincreasing}}
\end{figure}

Fig. \ref{Fig:Nonincreasing} verifies the cost function is non-increasing and converges. We plot the cost function in dB (blue line) and actual values (red line). The blue and red lines clearly show the non-increasing property and convergence of the proposed algorithm, respectively.

\subsection{Special Case: Nullforming Beampattern Design}
\label{Sec:NFB}

Null forming beampattern design can be seen as a special case of our full beampattern design. However, unlike the problem formulation in (\ref{Eq:BP}), the goal of null forming beampattern design is to form a beampattern with nulls in desired directions denoted by $\{ \theta_k\}^K_{k=1}$. Here, the objective function can be defined by

\begin{align}
f(\mb x) &= \sum_{p=-\frac{N}{2}}^{\frac{N}{2}-1} \|\mb A_p\mb W_p \mb x\|_2^2\\
% &=\sum_{k=1}^K \sum_{p=-\frac{N}{2}}^{\frac{N}{2}-1} |\mb a^H_{kp} \mb W_p \mb x|^2\\
&=\mb x^H \mb V \mb x
\end{align}
where $\mb V$ is expressed as
\be
\mb V=\sum_{p=-\frac{N}{2}}^{\frac{N}{2}-1} \mb W_p^H \mb A_p^H \mb A_p\mb W_p
\ee
Therefore, the minimization problem can be formulated as
\be
\label{Eq:BPNF}
\left\{ \begin{array}{cc}
	\displaystyle
	\min_{\mb x} &\mb x^H  \mb V \mb x\\
	%        \text{subject to:  } & |\mb x(k)|^2\leq 1/(MN) \\
	\text{s.t.:  } & |x_m(n)|=1, \text{ for }m=1, 2,\ldots, M \text{ and }\\
	 & \quad \quad \quad \quad n=0, 1,\ldots, N-1\\
	& \|\bar{\mb F}^H \bar{\mb y}-\mb x\|_2^2 \le E_R\\
\end{array} \right.
\ee
In this case, the optimization problem reduces to problem $CP^{(n)}$ in \eqref{Eq:BPCP} with $\mb R$ and $\mb s$ redefined as:
\be
\label{Eq:R_NB}
\mb R={\begin{bmatrix}
		\re \{\mb V\}   &-\im \{\mb V\}\\
		\im \{\mb V\}   &\re \{ \mb V\}\\
\end{bmatrix}}
\ee
\be
\mb s={\begin{bmatrix}
		\re \{\mb x\} \\
		\im \{\mb x\}  \\
\end{bmatrix}}
\ee
Since $\mb V$ is positive semi-definite and there are no linear terms in the objective function (i.e. $\mb q=\mb 0$ and $r= 0$ ), then all the lemmas in Section \ref{Sec:SCF} hold. Note that, in this case, we use only Algorithm \ref{Alg:BIC} with $\mb R$ as mentioned above.

\begin{algorithm}[!t]
	\caption{Beampattern optimization with spectral Interference control (BIC)}
	\label{Alg:BIC2}
	\begin{algorithmic}
		\State \tb{Inputs: } $d_{kp}$, $\mb W_p$, $\mb a_{kp}$, for $p= -\frac{N}{2},...,0, ..., \frac{N}{2}-1$, $k=1, 2, .., K$ and $\zeta$ (the desired threshold value).
		\State \tb{Output: } A solution $\mb x^{\star}$ for problem (\ref{Eq:BP}).
		% \State (1) Solve the QCQP of $RC^{(0)}$ using SOCP.
		\State (1) Set $m=1$.
		\State (2) Set $\phi^{(m)}_{kp}=\arg\{\mb a^H_{kp} \mb W_p \mb x^{(m-1)}\}$ for all $k$ and $p$.
		% \\(The computational complexity of this step is $\mathcal{O}({M}^{3.5}{N}^{3.5})$)
		\State (3) Set $\mb d^{(m)}_p=[d_{1p} e^{j\phi^{(m)}_{1p}}, ...,d_{Kp} e^{j\phi^{(m)}_{Kp}}]^T$.
		\State (4) Use Algorithm \ref{Alg:BIC} to compute $\mb x^{(m)}$ with $\mb d_p=\mb d^{(m)}_p$ and $\mb x^{(0)} (\text{in Algo. 1})=\mb x^{(m-1)}$ as inputs.
		% \\(The computational complexity of this step is $\mathcal{O}({M}^{2}{N}^{2})$ \cite{golub2012matrix})
		\State (5) Check the following:
		% \If{ $\sum_{p} \||\mb d_p|-|\mb A_p \mb W_p\mb x^{(m)}|\|_2^2-\sum_{p} \||\mb d_p|-|\mb A_p \mb W_p\mb x^{(m-1)}|\|_2^2<\zeta$}
		\If{ $f'(\mb x^{(m)})-f'(\mb x^{(m-1)})<\zeta$ where $f'(\mb x)=\sum_{p} \||\mb d_p|-|\mb A_p \mb W_p\mb x|\|_2^2$}
		\State STOP.
		\Else
		\State set $m=m+1$ GOTO step (2).
		\EndIf
		\State \tb{Output: }  $\mb x^{\star}=\exp\{j \arg(\mb x^{(m)})\}$.
		
	\end{algorithmic}
\end{algorithm}

%\begin{figure*}[!t]
%	\begin{center}
%		\subfloat[]{\includegraphics[scale=0.6]{figs/costfunction_lambda.pdf}\label{Fig:CR_CostFunction}}\hfil
%		\subfloat[]{\includegraphics[scale=0.6]{figs/energy_lambda.pdf}\label{Fig:CR_Energy}}\\
%	\end{center}
%	\caption{Rate of Convergence as a function of $\lambda$}\label{Fig:ConvergenceRate}
%\end{figure*}

\begin{figure*}[!t]
	\begin{center}
		\subfloat[beampattern vs. angle]{\includegraphics[scale=0.67]{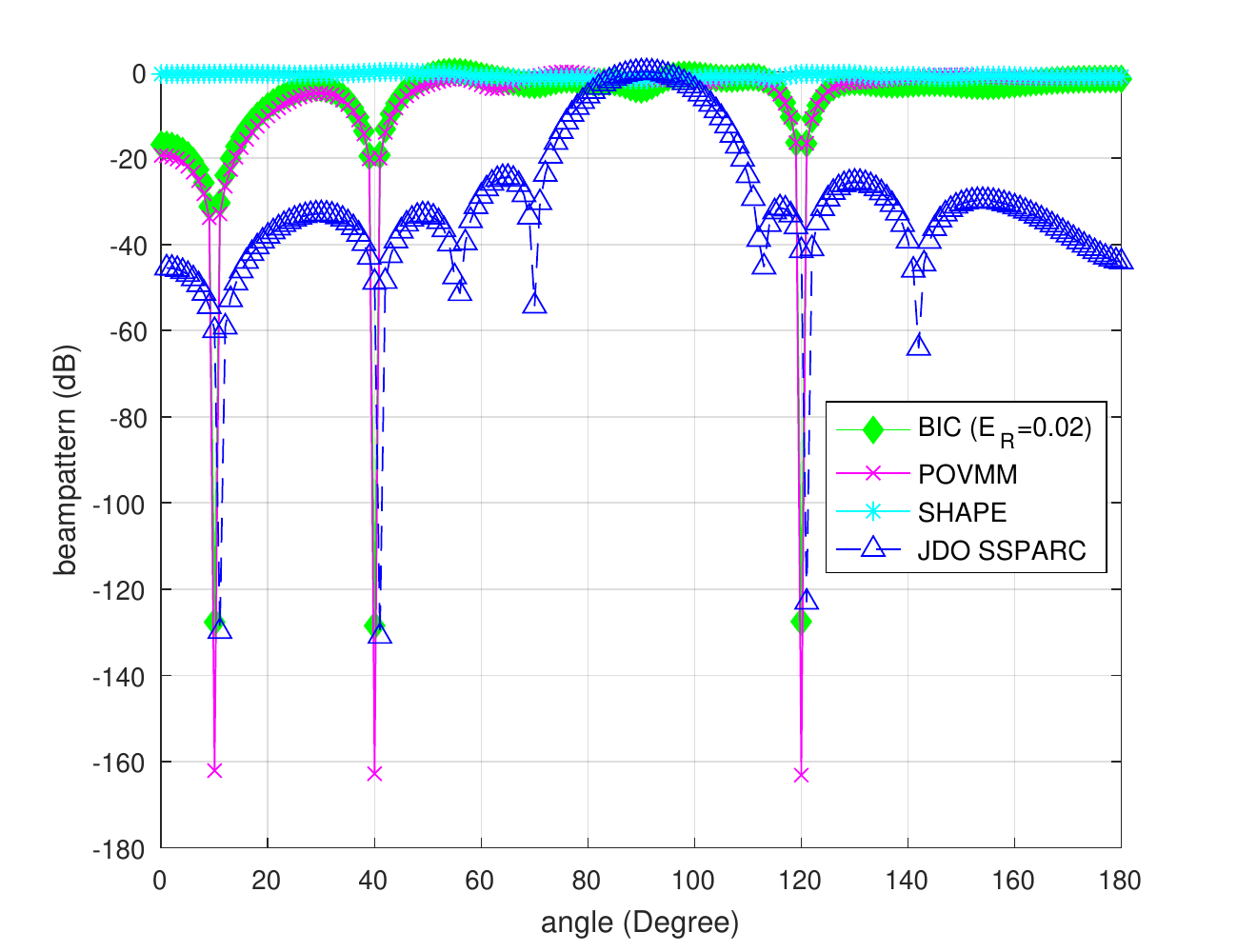}\label{Fig:CBICPO}}\hfil
		\subfloat[spectrum vs. frequency]{\includegraphics[scale=0.6]{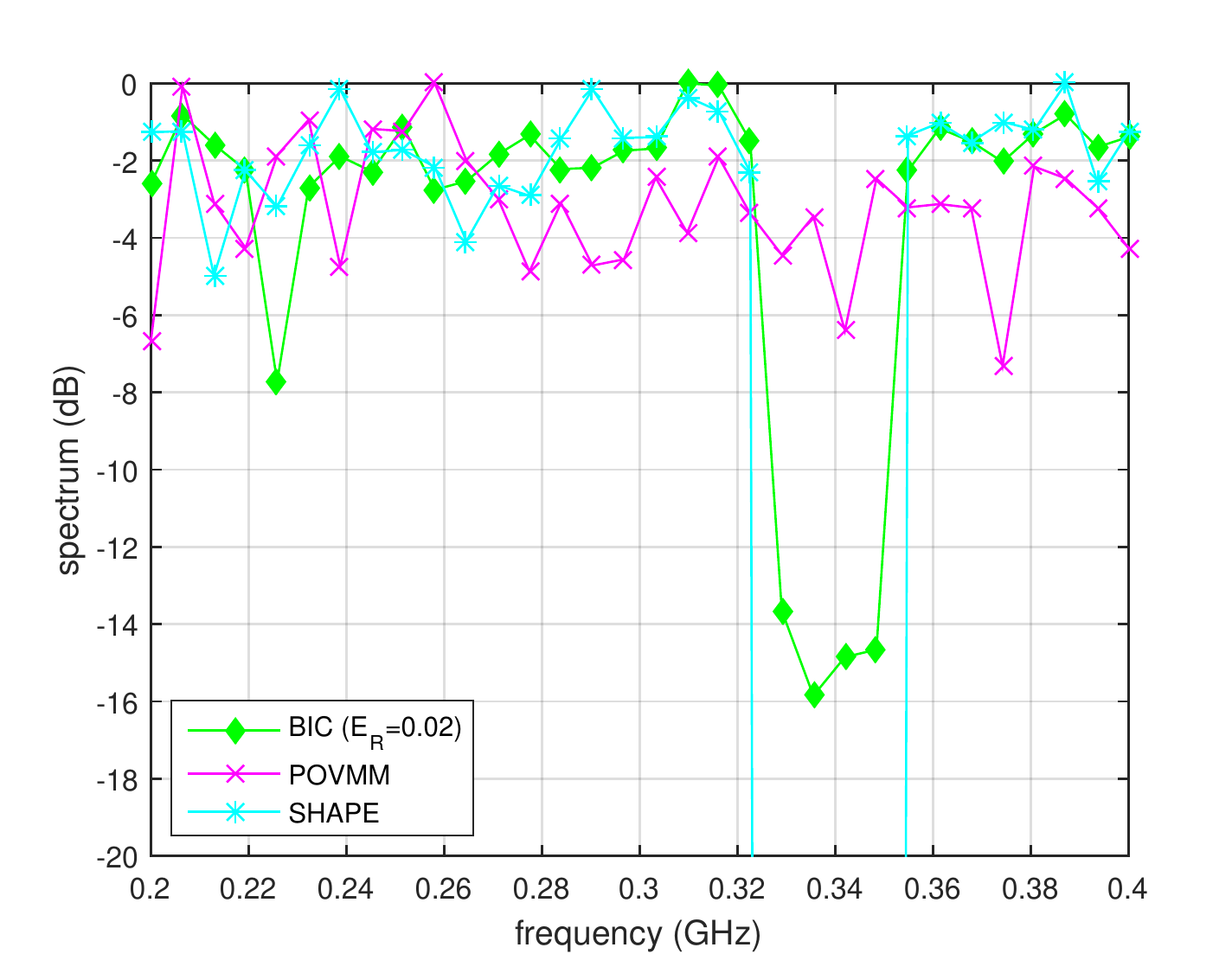}\label{Fig:CBICPOS}}\\
		\subfloat[spectrum in TV bands]{\includegraphics[scale=0.6]{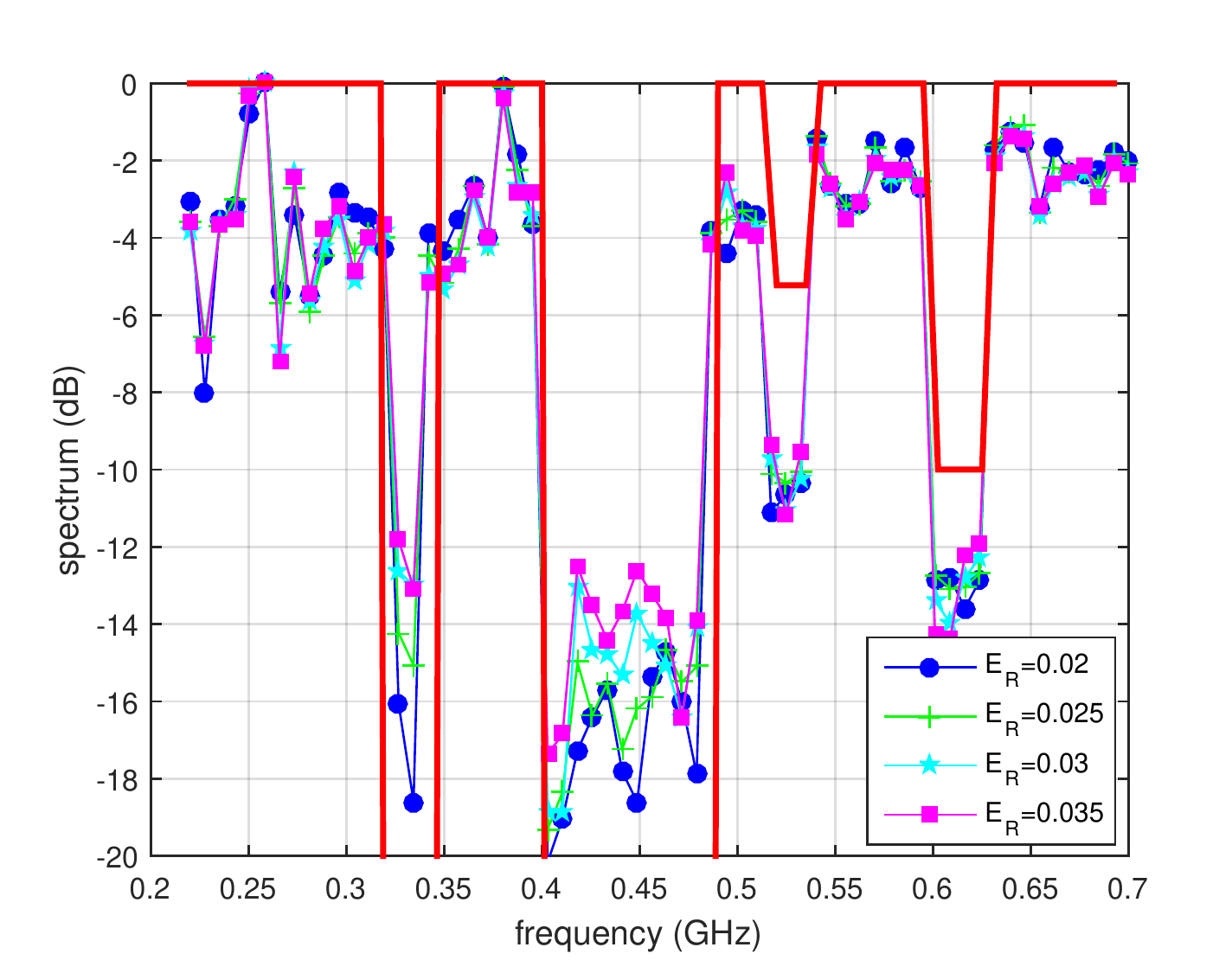}\label{Fig:TVBand}}\hfil
		\subfloat[value of cost function vs. $E_R$]{\includegraphics[scale=0.6]{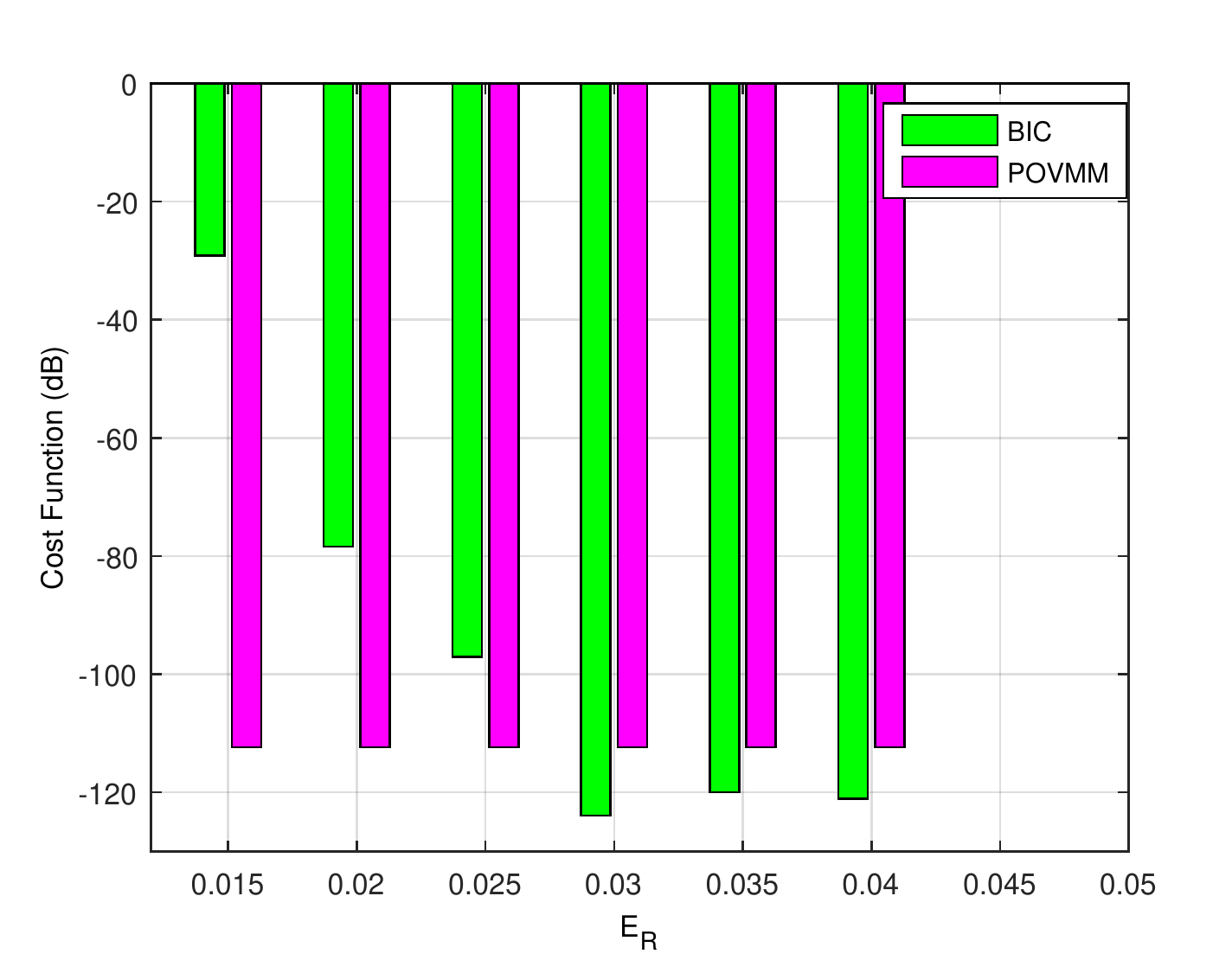}\label{Fig:CostBIC}}
	\end{center}
	\caption{Nullforming beampattern design}\label{Fig:Nullforming}
\end{figure*}

\section{Numerical Results}
\label{Sec:Results}
%In this section, we first invesigate the convergence rate of the BIC algorithm and evaluate how $\lambda$ in the cost function of \eqref{Eq:BPreal} affects the convergence rate. Then 

We examine the performance of the proposed BIC by comparing it against the following well-known methods:
%
%
%\subsection{Convergence Rate Analysis}
%We plot the value of cost function $f(\mb x)$ and $\mb x^H \mb x - L$ versus the iteration index with different values of $\lambda$ in Fig. \ref{Fig:ConvergenceRate}. Fig. \ref{Fig:CR_CostFunction} shows that the smaller $\lambda$ not only generates the smaller cost function at convergence, which means that the optimized beampattern is closer to the ideal beampattern, but also results in faster convergence. However, Fig. \ref{Fig:CR_Energy} tells us that a large $\lambda$ ($\geq 100$) generates the constant modulus waveform and waveform from small $\lambda$ departs from constant modulus. Note that though $\mb x^H \mb x$ is the energy of the waveform, this can be used to see whether constant modulus is achieved or not because $\mb x^H \mb x \geq L$ always holds at every iteration step. From Figs. \ref{Fig:CR_CostFunction} and \ref{Fig:CR_Energy}, $\lambda$ can be chosen appropriately according to which is more important, fidelity to the desired beampattern and fast convergence or constant modulus.
%
%\subsection{Methods of Comparison}
%We evaluate and compare the proposed BIC against state-of-the-art algorithms for waveform design and beampattern optimization under a constant modulus constraint. 

%Note that none of the following algorithms do not exploit the spectral interference constraint.

\begin{itemize}
	\item \textbf{Phase-only variable metric method (POVMM) \cite{Guo15}:} POVMM performs null forming beampattern design by optimizing phases of the waveform under the constant modulus constraint but no spectral constraint is involved.
	\item \textbf{SHAPE \cite{Rowe14}:} The SHAPE algorithm is a computationally efficient method of designing sequences with desired spectrum shapes. In particular, the spectral shape is optimized as a cost function subject to the constant modulus constraint but the resulting beampattern is an outcome (not explicitly controlled).
	\item \textbf{JDO SSPARC \cite{Guerci15}:} An approach for beamforming that maximizes the signal power through the forward channels while simultaneously minimizes the response at the co-channels. Note that, JDO SSPARC does not control the spectral shape of the waveform in the frequency domain.
	%	\item \textbf{SDR with randomization \cite{Cui14TSP}:} Cui \textit{et al.} proposed the MIMO waveform design algorithm maximizing the signal-to-interference-plus-noise ratio (SINR) under the constant modulus constraint. To solve the constrained optimization problem in a tractable manner, they employed SDR with randomization. Though the algorithm in \cite{Cui14TSP} is for SINR maximization, it is also applicable to the beampattern design since the cost function has a quadratic form.
	\item \textbf{Wideband beampattern formation via iterative techniques (WBFIT) \cite{He11}:} The WBFIT synthesize wideband MIMO beampattern under the constant modulus or low PAR. They first find the Fourier transformed waveform in the frequency domain and then fit the DFT of the waveform to the result of the first step subject to the enforced PAR constraint. 
\end{itemize}

\textit{Remark:} The initial sequence (waveform code) adopted in the numerical results is a pseudo-random sequence of unit magnitude entries. The proposed algorithm  consistently converges to a lower objective function value regardless of the initial sequence.

\subsection{Nullforming Beampattern Design}
\label{Subsec:NBE}

\noindent We compare  BIC to state-of-the-art phase-only variable metric method (POVMM) method \cite{Guo15} and the SHAPE algorithm \cite{Rowe14}. The experimental set up is as follows: We simulate a linear MIMO radar antenna array of $M=16$ elements with half-wavelength spacing and number of time samples $N=32$. In Algorithm 1 and 2, we set $\zeta = 10^-5$. Further, $K=3$, $\theta=[10^{\circ}, 40^{\circ}, 120^{\circ}]$. We assume a carrier frequency of $f_c=300$ MHz and allowed access to the 225-328.6 MHz and 335-400.15 MHz bands allocated for the U.S. Federal Government. We then place a notch in the band 328.6-335 MHz.

\begin{figure}[!t]
	\centering
	\subfloat[]{\includegraphics[scale=0.6]{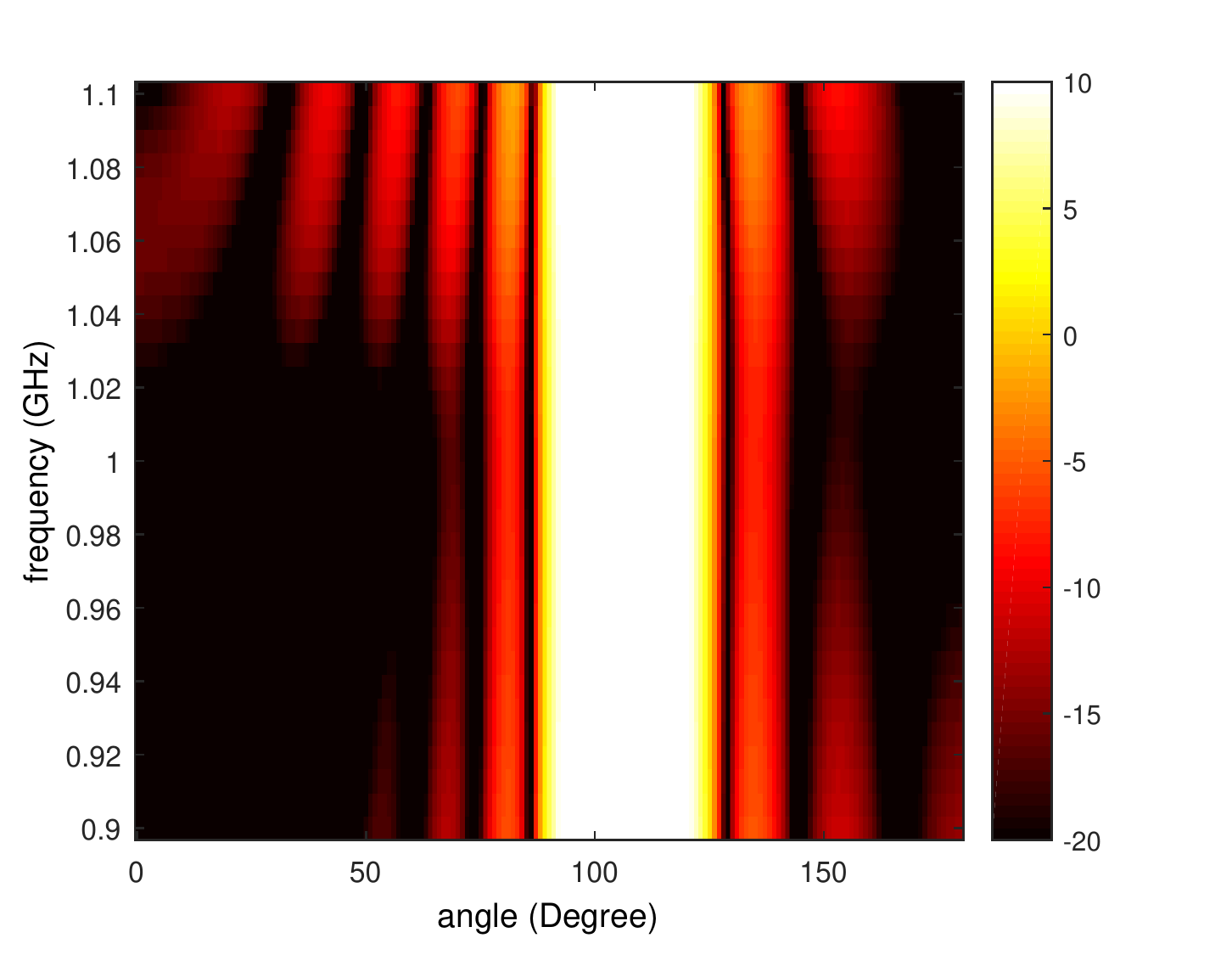}\label{Fig:95120unconstrained}}\\
	\subfloat[]{\includegraphics[scale=0.6]{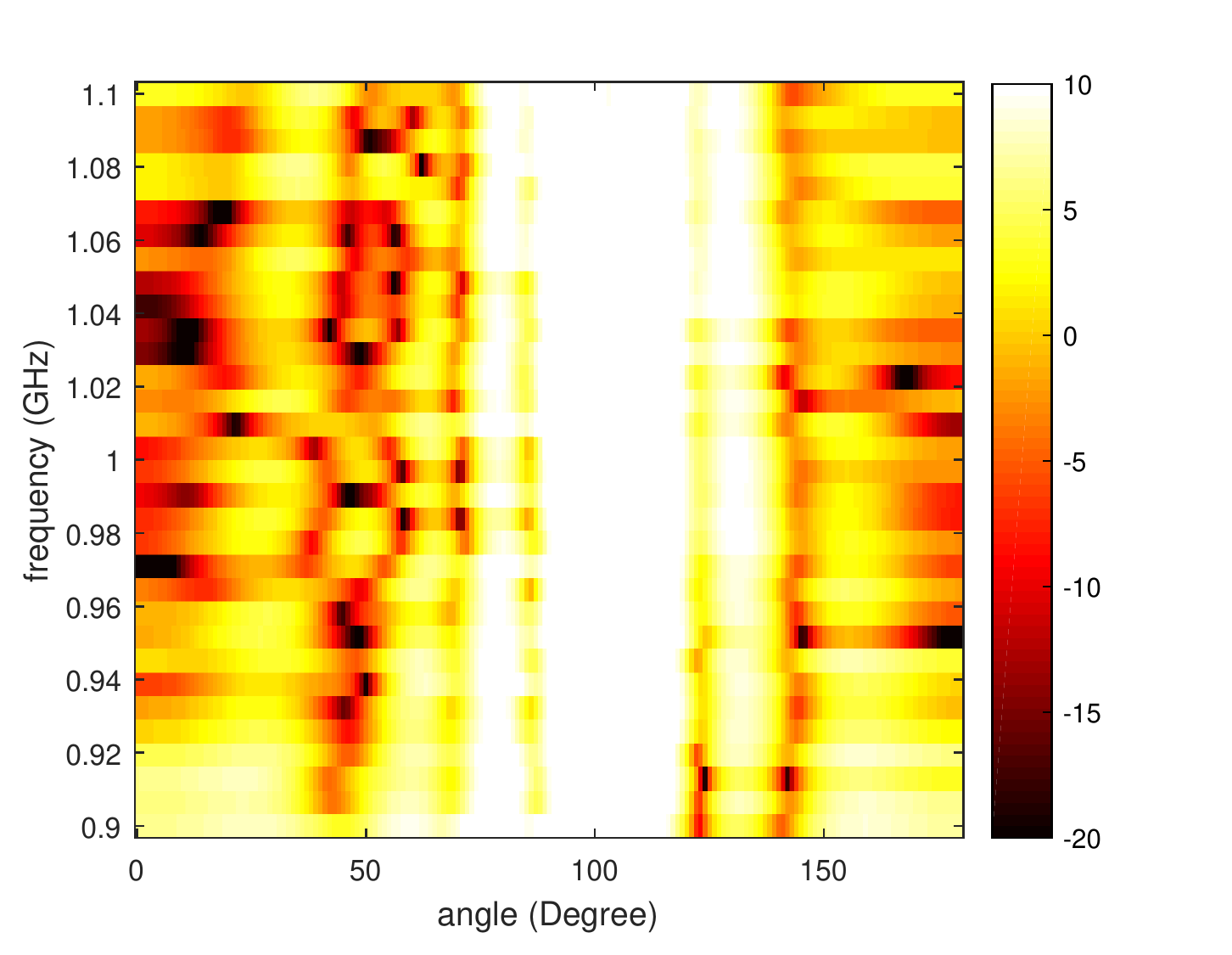}\label{Fig:95120WBFIT}}\\
	\subfloat[]{\includegraphics[scale=0.6]{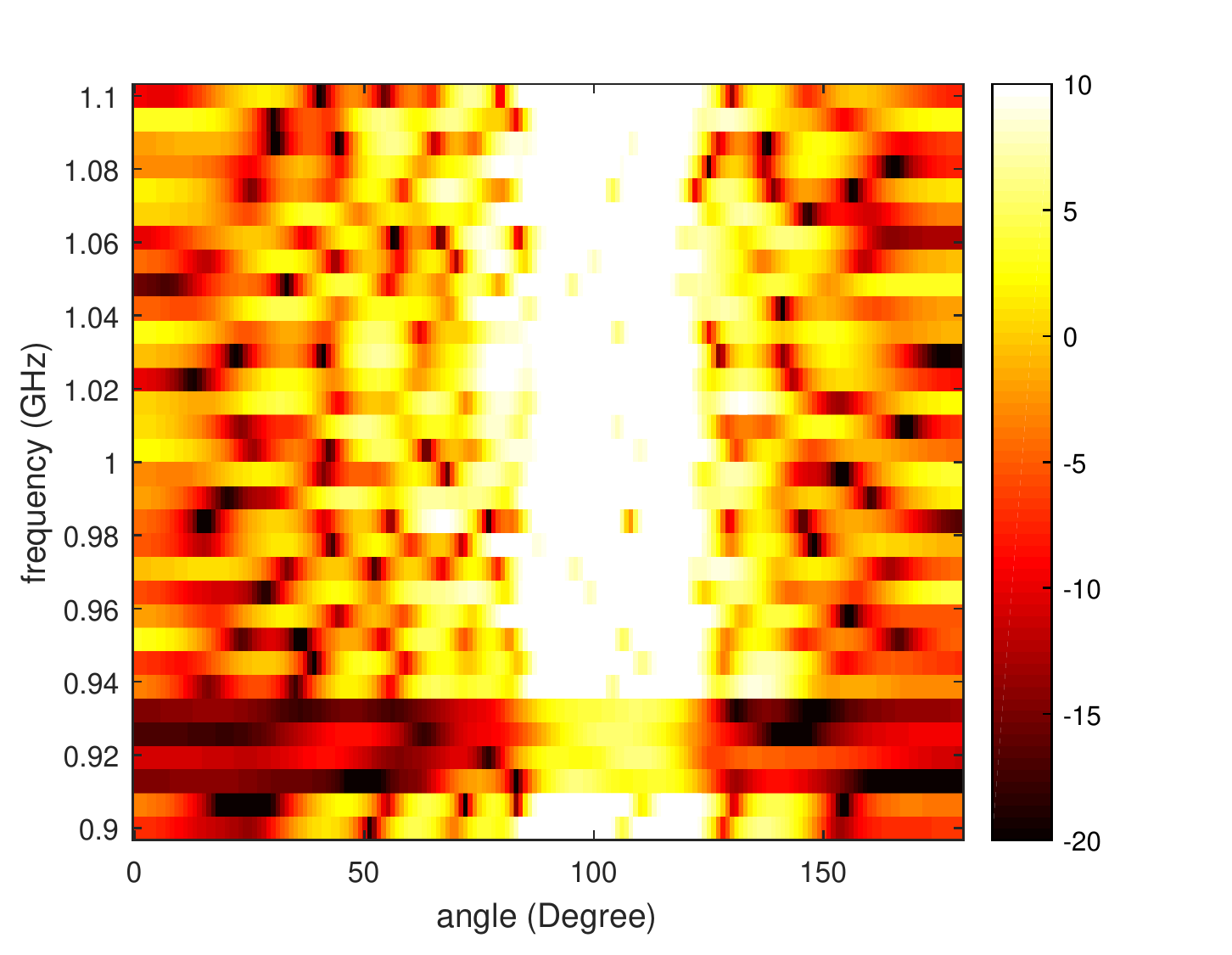}\label{Fig:95120BIC}}\\
	\caption{Plot of the beampattern. (a) unconstrained (b) WBFIT method (c) BIC (Proposed method)} \label{Fig:95120}
\end{figure}

Fig. \ref{Fig:Nullforming} shows the results for nullforming beampattern of BIC versus POVMM and SHAPE. Fig. \ref{Fig:CBICPO}, we plot the resulting beampattern versus the angle. Each of BIC, POVMM, JDO SSPARC  achieve nulls in the desired angles, i.e. desired spatial control. SHAPE lacks a spatial control component by virtue of its design. Note that the forward channel for JDO SSPARC is set to be  $\theta=[80^{\circ} \text{ to } 100^{\circ}]$, however, unlike the other methods, the resulting waveform is non-constant modulus. On the other hand, Fig. \ref{Fig:CBICPOS} plots the spectrum versus the frequency. Here, BIC and SHAPE effectively suppress the energy in the frequency bands where the transmission should be mitigated.  Unsurprisingly, POVMM do not provide the desired suppression in the frequency bands of interest because it is not designed for the same. In summary, only the proposed BIC enables the desired spatio-spectral control. 

In Fig. \ref{Fig:TVBand}, we investigate a more practical scenario. We assume we have access to licensed television braodcasts (UHF) that occur from 470 to 698 MHz as well as the 225-328.6 MHz and 335-400.15 MHz bands as in Fig. \ref{Fig:CBICPO}. Each television station is allocated 6 MHz of bandwidth and we assume there are 7 stations are licensed for operation (Ch. 21-23, 512-536 MHz and Ch. 36-39, 602-626 MHz). We plot the spectrum as achieved by different methods with different threshold ($E_{R}$) values in Fig. \ref{Fig:TVBand} and as expected a smaller threshold ($E_{R}$ value) leads to a tighter spectral constraint. 

It is also shown in Fig. \ref{Fig:TVBand} that the spectral constraint can be set to incorporate the information of the distance of a TV station/wireless interferer to the radar . In particular the results in Fig. \ref{Fig:TVBand} assume that the stations of Ch. 36-39 are closer to the radar than Ch. 21-23.  $\bar{\mathbf y}$ in (\ref{Eq:BPNF}) is appropriately set (see red curve in Fig. \ref{Fig:TVBand}) to control the relative importance of frequency bands. 

In Fig. \ref{Fig:CostBIC}, we show the cost function value corresponding to POVMM and the proposed BIC (recall, they optimize the same cost function in the nullforming case). The BIC method achieves similar cost function values or lower when $E_R\ge 0.03$. This is particularly remarkable because BIC additionally enforces the spectral constraint. Finally, the performance of the proposed BIC method in terms of the total normalized interference energy as well as average spatial cancellation in the three nulls is shown in Table \ref{Tb:IESC}.

\begin{table}[h]
	\begin{center}
		\caption{total interference energy (IE) versus average spatial cancellation (SC)}\label{Tb:IESC}
		\begin{tabular}{|c|c|c|}
			\hline
			\textbf{Method} & Normalized IE &Average SC (dB)\\\hline
			POVMM & 0.124 &166.4\\\hline
			BIC ($E_R = 0.02$) &  0.007 &-127.83\\\hline	
			BIC ($E_R = 0.025$) & 0.0092& -169.2\\\hline	
			BIC ($E_R = 0.03$) & 0.0114& -209.7\\\hline
			BIC ($E_R = 0.03$) & 0.0198& -280.6\\\hline
		\end{tabular}
	\end{center}
	
\end{table}

\subsection{Full Beampattern Design}
\label{Subsec:WBE}
For wideband beampattern design, we compare BIC to the state-of-the-art WBFIT method \cite{He11}. The experimetal set-up used in Fig. \ref{Fig:95120} and Fig. \ref{Fig:Box} is following. The number of transmit antennas $M=10$, the number of time samples $N=32$, the carrier frequency of the transmit signal $f_c=1$ GHz and the bandwidth $B=200$ MHz and the spatial angle is divided into $K=180$ grid points.

In Fig. \ref{Fig:95120}, we place a notch in the band 910-932 MHz and consider the following desired transmit beampattern

\begin{equation}
d(\theta,f)=
\begin{cases}
1 & \theta=[95^{\circ}, 120^{\circ}]\\
0 & \text{Otherwise}.
\end{cases}
\end{equation}
Fig. \ref{Fig:95120} shows the angle-frequency plot of the beampattern for WBFIT method (no spectral constraint) and BIC with the spectral  constraint ($E_R=0.01$). The BIC method is able to keep the energy of the waveform in particular frequency band low enough as well as achieve higher suppression at the undesired angles compared to WBFIT.

In Fig. \ref{Fig:Box}, we simulate a more challenging practical scenario. We assume that the beampattern should be suppressed at the angles of 40$^\circ$ through 80$^\circ$ in the frequency band [943.75 MHz, 981.25 MHz] and at 120$^\circ$ through 160$^\circ$ in [962.5 MHz, 1,000 MHz], that is,
\begin{equation}
\label{eq:specwidesetup}
d(\theta,f)=
\begin{cases}
0 & \theta=[40^{\circ}, 80^{\circ}] \text{ and } f=[943.75, 981.25]\\
0 & \theta=[120^{\circ}, 160^{\circ}] \text{ and } f=[962.5, 1000]\\
1 & \text{Otherwise}.
\end{cases}
\end{equation}
This ideally appears as black boxes in the angle-frequency beampattern plots. We also assume that transmission should be restricted at all directions in the frequency band [1.025 GHz, 1.0625 GHz]. This restriction can be performed by the spectral constraint. First, as shown in Fig. \ref{Fig:boxWBFIT}, since WBFIT does not have the spectral constraint, the notch of frequency band [1.025 GHz, 1.0625 GHz] does not appear. Second, the black boxes are not seen so clearly in Fig. \ref{Fig:boxWBFIT}. Lastly, WBFIT suppresses the energy of the waveform unnecessarily in the frequency band where we do not have any restriction (e.g. [1.0625 GHz, 1.1 GHz]). On the other hand, the proposed BIC effectively suppresses and restricts the transmitted energy in the desired frequency bands and angles and generate enough power elsewhere.

\begin{table}[h]
	\begin{center}
		\caption{Converged cost function values in dB}\label{Tb:CostFunction}
		\begin{tabular}{|m{3cm}|m{3cm}|}
			\hline
			\textbf{Method} & cost function (dB)\\\hline
			Unconstrained & 15.4681\\\hline
			WBFIT & 34.7744\\\hline
			BIC ($E_R = 0.01$) & 32.6461\\\hline	
			BIC ($E_R = 0.02$) & 31.3286\\\hline	
			BIC ($E_R = 0.03$) & 30.8468\\\hline
		\end{tabular}
	\end{center}
	
\end{table}

Finally, we compare values of the cost function of each algorithm for the same scenario in (\ref{eq:specwidesetup}) and the results are reported in Table \ref{Tb:CostFunction}. In Table \ref{Tb:CostFunction}, unconstrained wideband beampattern design (not even a constant modulus constraint) plays the role of a lower bound. BIC outperforms WBFIT even as it incorporates an additional spectral constraint.

\begin{figure}[!t]
	\centering
	\subfloat[]{\includegraphics[scale=0.6]{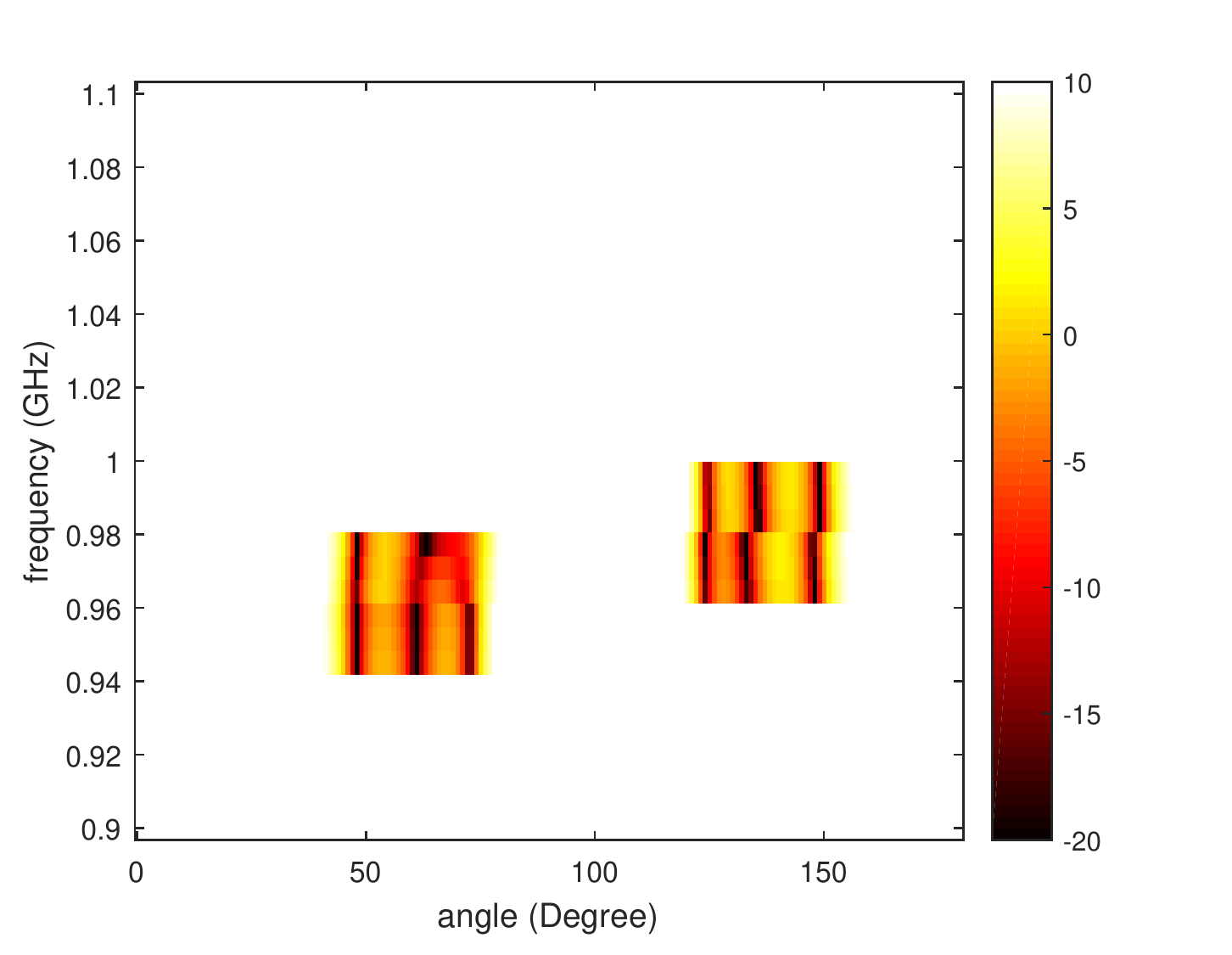}\label{Fig:boxunconstrained}}\\
	\subfloat[]{\includegraphics[scale=0.6]{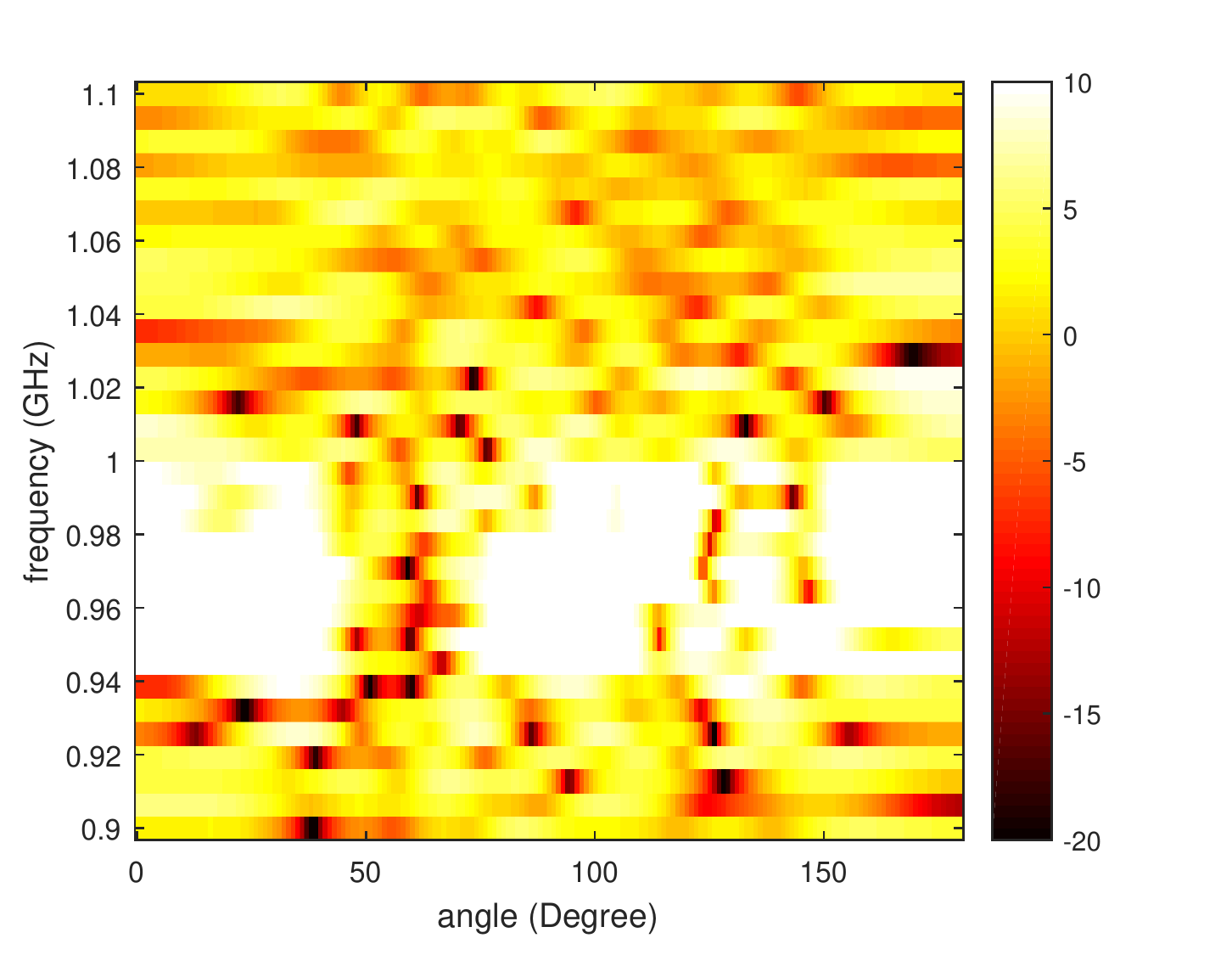}\label{Fig:boxWBFIT}}\\
	\subfloat[]{\includegraphics[scale=0.6]{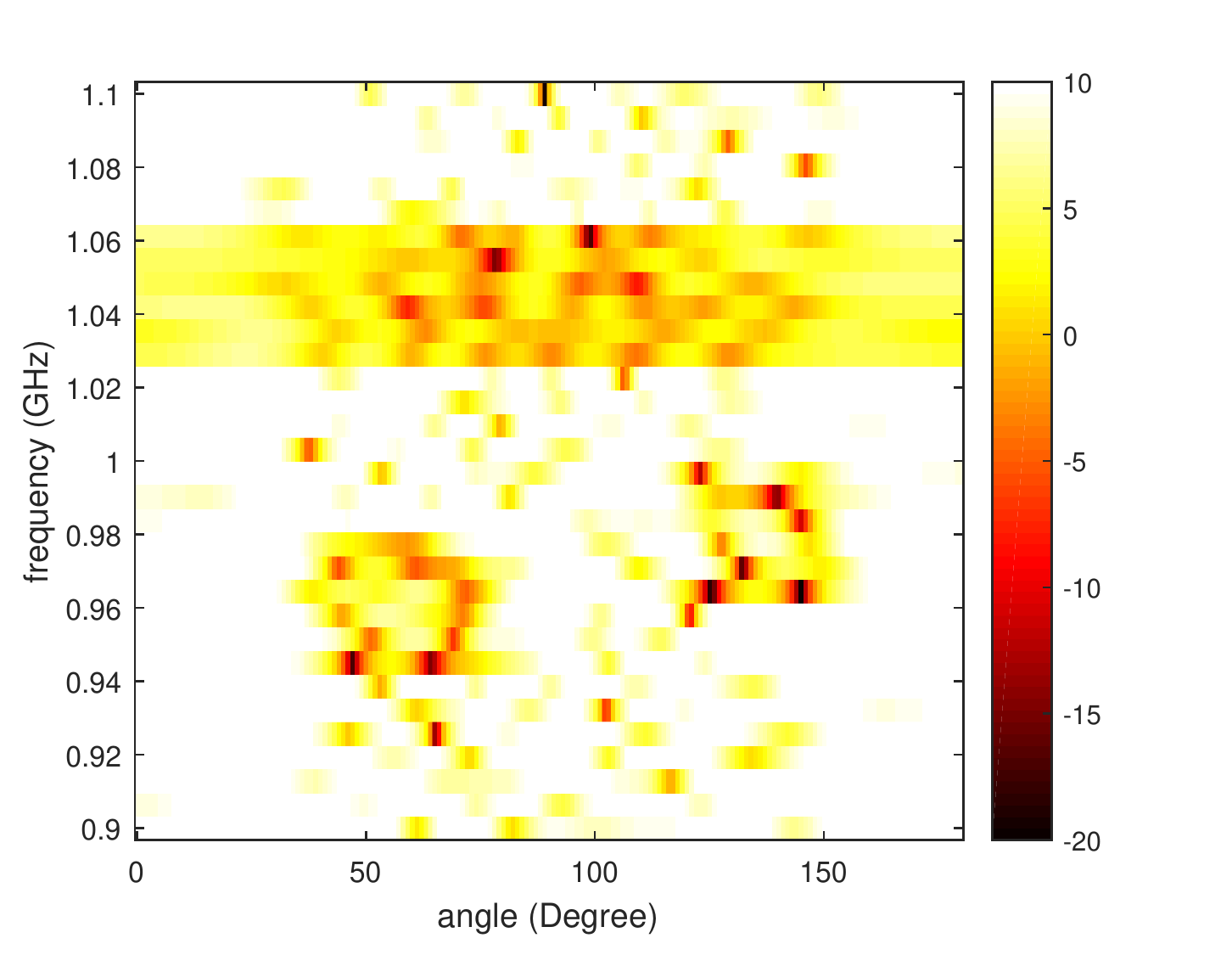}\label{Fig:boxBIC}}\\
	\caption{Plot of the beampattern. (a) unconstrained (b) WBFIT method (c) BIC (Proposed method)} \label{Fig:Box}
\end{figure}

Table \ref{TableII} shows computational complexity and run times as observed in the simulation. Note that, both POVMM and WBFIT do not have a spectral constraint unlike the proposed BIC method, hence, they have a computational advantage over the proposed BIC method. However, although POVMM has lower complexity per iteration, it needs more iterations to achieve the same performance as BIC for high $E_r$ values. For a fair comparison, BIC as well as competing methods are initialized with the same waveform, which is a psuedo-random vector of unit magnitude complex entries.

  \begin{table}[ht]
  \caption{Computational complexity for different methods} \label{TableII}
  \centering
  \begin{tabular}{|c|c|c|c|}
  \hline
  \textbf{Method}  & \textbf{Sim. Time (s)} & \textbf{iter.}&\textbf{Comp. order}\\ \hline
    WBFIT              & $0.7106$             & 80&$\mathcal{O}(FN{M}^{2})$\\ \hline
   POVMM              & $11.5081$             & 600&$\mathcal{O}(F{L}^{2})$\\ \hline
  BIC ($\zeta=10^{-11}$,$E_R = 0.4$)  & $9.5239$     & 25&$\mathcal{O}(F{L}^{2.373})$\\ \hline
  \end{tabular}
  \end{table}

\section{Conclusion}
\label{Sec:Conclusion}

Our work achieves tractable spatio-spectral beampattern design by waveform optimization for MIMO radar in the presence of constant modulus and spectral constraints. The central idea of our analytical contribution is to successively achieve constant modulus (at convergence), while solving a quadratic program with linear equality and inequality constraints in each step of the sequence. Because each problem in the sequence has a closed form, this makes our method computationally attractive. We establish new analytical properties of the BIC algorithm such as non-increasing cost function in each iteration and guaranteed convergence. Further, we show experimentally that the proposed BIC can achieve superior beampattern accuracy compared to many state-of-the-art methods even as BIC solves a spectrally constrained problem. Future work could consider the incorporation of additional constraints such as waveform similarity \cite{Friedlander07,Chen09} and explore further optimality properties of the BIC solution.

%%%% References
%\scriptsize
%\bibliographystyle{C:/Users/BXK265/BOXSYN\string~2/BOXSYN\string~1/LaTex/Bibliography/IEEEbib}
%\bibliography{C:/Users/BXK265/BOXSYN\string~2/BOXSYN\string~1/LaTex/Bibliography/IEEEabrv,C:/Users/BXK265/BOXSYN\string~2/BOXSYN\string~1/LaTex/Bibliography/Bosung}
\bibliographystyle{IEEEbib}
\bibliography{IEEEabrv,Bosung}

\end{document}